%% file: main.tex
\documentclass[11pt, a4paper]{article}
\usepackage[margin=1in]{geometry}
\input{commands}

\title{Tree decompositions meet induced matchings:\\ beyond Max Weight Independent Set\thanks{Lima is supported by the Independent Research Fund Denmark grant agreement number 2098-00012B. Milani{\v c}, Mur{\v s}i{\v c}, and \v Storgel are supported in part by the Slovenian Research and Innovation Agency (Milani{\v c}: I0-0035, research program P1-0285 and research projects J1-3001, J1-3002, J1-3003, J1-4008, J1-4084, and N1-0102; Mur{\v s}i{\v c}: research project J1-4008; \v Storgel: research program P1-0383 and research projects  J1-3002 and J1-4008).
Milani{\v c} is also supported in part by the research program CogniCom (0013103) at the University of Primorska.
Okrasa is supported by the Polish National Science Centre grant no. 2021/41/N/ST6/01507.
Rzążewski is supported by the European Research Council (ERC) under the European Union’s Horizon 2020 research and innovation programme grant agreement number 948057.
}}

\author{
Paloma T.\ Lima\\
\small IT University of Copenhagen, Denmark\\
\small \texttt{palt@itu.dk}
\and
Martin Milani\v c\\
\small FAMNIT and IAM,\\
\small University of Primorska, Slovenia\\
\small \texttt{martin.milanic@upr.si}
\and
Peter Mur\v si\v c\\
\small FAMNIT, University of Primorska, Slovenia\\
\small \texttt{peter.mursic@famnit.upr.si}
\and
Karolina Okrasa\\
\small Warsaw University of Technology, Poland\\
\small \texttt{karolina.okrasa@pw.edu.pl}
\and
Pawe\l{} Rz\k{a}\.zewski\\
\small Warsaw University of Technology and\\
\small University of Warsaw, Poland\\
\small \texttt{pawel.rzazewski@pw.edu.pl}
\and
Kenny \v Storgel\\
\small Faculty of Information Studies and\\
\small FAMNIT, University of Primorska, Slovenia\\
\small \texttt{kennystorgel.research@gmail.com}
}
\date{}

\begin{document}
\begin{titlepage}
\maketitle

\begin{abstract}
For a tree decomposition $\mathcal{T}$ of a graph $G$, by $\mu(\mathcal{T})$ we denote the size of a largest induced matching in $G$ all of whose edges intersect one bag of $\mathcal{T}$.
\emph{Induced matching treewidth} of a graph $G$ is the minimum value of $\mu(\mathcal{T})$ over all tree decompositions $\mathcal{T}$ of~$G$.
Yolov~[SODA 2018] proved that \textsc{Max Weight Independent Set} can be solved in polynomial time for graphs of bounded induced matching treewidth.

In this paper we explore what other problems are tractable in such classes of graphs.
As our main result, we give a polynomial-time algorithm for \textsc{Min Weight Feedback Vertex Set}.
We also provide some positive results concerning packing induced subgraphs, which in particular imply a PTAS for the problem of finding a largest induced subgraph of bounded treewidth.

These results suggest that in graphs of bounded induced matching treewidth, one could find in polynomial time a maximum-weight induced subgraph of bounded treewidth satisfying a given \textsf{CMSO}$_2$ formula.
We conjecture that such a result indeed holds and prove it for graphs of bounded \emph{tree-independence number}, which form a rich and important family of subclasses of graphs of bounded induced matching treewidth.

We complement these algorithmic results with a number of complexity and structural results concerning induced matching treewidth.
\end{abstract}
\def\thepage{}
\thispagestyle{empty}
\end{titlepage}

 \newpage
 \tableofcontents
 \newpage
 \clearpage

\section{Introduction}
Structured graph decompositions and the corresponding graph width parameters have become one of the central tools for dealing with algorithmically hard graph problems.
One of the  best known and well-studied graph width parameters is \emph{treewidth}, which was {introduced} independently by several groups of authors with various motivations~\cite{bertele1972nonserial,Halin1976,DBLP:journals/jct/RobertsonS84,Arnborg1987}.
Roughly speaking, treewidth measures how similar the graph is to a tree.
A graph $G$ with small treewidth can be represented by a tree $T$ whose every node holds a small subset (called {\emph{bag}}) of vertices of $V(G)$, such that the connectivity properties of $G$ are reflected in the local structure of $T$ (and the bags).
A tree decomposition $\cT$ consists of the tree $T$ and a function that assigns a bag to each node of $T$,
and the \emph{width} of $\cT$ is the size of its largest bag minus 1 (this is by convention).
The \emph{treewidth} $\tw(G)$ is the minimum width of a tree decomposition of $G$.

Having a tree decomposition of small width is excellent for algorithmic applications, as we can mimic the standard bottom-up dynamic programming on trees.
For example, in order to solve \textsc{Max Weight Independent Set} (\MWIS), the dynamic programming table for each node of $T$ should be indexed by all possible independent sets in the corresponding bag.
Consequently, if each bag is \emph{small}, say, its size is bounded by a constant $k$, the running time we obtain is {$2^k\cdot (n+|V(T)|)^{\mathcal{O}(1)}$ where $n \coloneqq |V(G)|$.}
As $|V(T)|$ can also be assumed to be polynomial in~$n$, this yields a polynomial-time algorithm for \MWIS{} {when restricted to} graphs of bounded treewidth.

Such an approach {has} found numerous applications in solving classic \NP-hard problems on graphs with bounded treewidth; see the handbook by Cygan et al.~\cite[Section 7]{platypus}.
The richness of the family of problems that can be solved by a dynamic programming over a tree decomposition is witnessed by the celebrated meta-theorem by Courcelle~\cite{DBLP:journals/iandc/Courcelle90}. It asserts that each problem expressible in \emph{Counting Monadic Second Order Logic} (\textsf{CMSO$_2$})%
\footnote{In this logic one can use vertex, edge, and (vertex or edge) set variables,
check vertex-edge incidence, quantify over variables,
and apply counting predicates modulo fixed integers. See~\cref{sec:cmso} for a formal introduction.}%
can be solved in polynomial time (actually, in linear time) {for} graphs of bounded treewidth.

An astute reader might notice a caveat in the results mentioned above: they assume that the input graph is given with a corresponding decomposition. This might be a serious obstacle, as computing an optimal tree decomposition is \NP-hard~\cite{Arnborg1987,DBLP:journals/corr/abs-2301-10031}.
However, if treewidth is bounded by a constant, an optimal decomposition can be found in polynomial time~\cite{DBLP:journals/siamcomp/Bodlaender96,DBLP:conf/stoc/KorhonenL23}.
Actually, in the context discussed above even an approximation of an optimal decomposition is sufficient.
Luckily, such approximation algorithms are known not only for treewidth~\cite{DBLP:journals/jct/RobertsonS95b,DBLP:journals/siamcomp/BodlaenderDDFLP16,DBLP:conf/focs/Korhonen21}, but also for other parameters discussed in the paper~\cite{DBLP:conf/soda/Yolov18,dallard2022computing}.
Thus, we will implicitly assume that the instance graphs are always provided with a corresponding tree decomposition.

Let us go back to the \MWIS problem, which serves as the starting point of our investigations.
Notice that in the argument sketched above we do not really need the fact that the treewidth is bounded by a constant.
Indeed, if the size of each bag is bounded by a \emph{logarithmic function of $n$}, then the number of all independent sets in a single bag is bounded by $2^{\Oh(\log n)} = n^{\Oh(1)}$, i.e., by a polynomial, which still yields a polynomial-time algorithm for \MWIS. This is one of the motivations to study classes of graphs with \emph{logarithmic treewidth},
which has been a very active topic in structural graph theory in recent years~\cite{threepath,DBLP:conf/soda/BonamyBDEGHTW23}.

However, the crux of the algorithm above is not really the size of the bag, but {the number of independent sets inside each bag}.
Consider, for example, the class of chordal graphs, i.e., graphs that do not contain an induced cycle with at least four vertices. Equivalently, these are graphs that admit a tree decomposition whose every bag is a clique.
Even though these cliques can be arbitrarily large, they still contain very few (that is, polynomially many) independent sets and thus such a decomposition can be used to solve \MWIS in polynomial time.

This observation is heavily extended by the framework of \emph{potential maximal cliques} by Bouchitt\'e and Todinca~\cite{BouchitteT01}. Essentially, there we work with an implicitly given tree decomposition whose every bag is promised to contain only a constant number of vertices from the optimal independent set.
This approach allows us, for example, to solve \MWIS and many related problems {for} graph classes with polynomially many minimal separators~\cite{BouchitteT01, FominTV15}.
With some care, it can also be used for other graph classes, e.g., graphs that exclude a fixed induced path~\cite{LokshantovVV14,p6FreeMaxInd22} or long induced cycles~\cite{DBLP:conf/soda/AbrishamiCPRS21}.

Another, more direct way of generalizing bounded-treewidth graphs and chordal graphs was suggested by Yolov~\cite{DBLP:conf/soda/Yolov18} and later, independently, by Dallard, Milani\v{c}, and \v{S}torgel~\cite{dallard2022firstpaper}. We use the notation and terminology of Dallard et al.~\cite{dallard2022firstpaper}.
Given a tree decomposition $\cT$ of a graph $G$, by $\alpha(\cT)$ we denote the size of a largest independent set contained in a bag of $\cT$. The \emph{tree-independence number} of $G$, denoted by $\tin(G)$, is the minimum value of $\alpha(\cT)$ over all tree decompositions $\cT$ of $G$.
Note that given a tree decomposition $\cT$ of $G$ with $\alpha(\cT) \leq k$, we can solve \MWIS on $G$ in time $n^{\Oh(k)}$, which is polynomial for constant $k$.
Let us remark that we always have $\tin(G) \leq \tw(G)+1$ and if $G$ is chordal, then $\tin(G) \leq 1$ (actually, the reverse implication also holds~\cite{dallard2022firstpaper}).
Tree-independence number proved to be a fruitful topic in recent years and attracted some attention, both with structural~\cite{dallard2022secondpaper,abrishami2023tree} and algorithmic~\cite{dallard2022firstpaper,dallard2022computing} motivations.

In a recent manuscript, Dallard, Fomin, Golovach, Korhonen, and Milani\v{c}~\cite{dallard2022computing} showed that tree-independence number is in a sense the {most general} ``natural'' parameter of a tree decomposition whose boundedness yields a polynomial-time algorithm for \MWIS.
More precisely, let $\gamma$ be any graph invariant satisfying $\gamma(G - v)  \leq \gamma(G) \leq \gamma(G-v)+1$ for any graph $G$ and any $v \in V(G)$.
For a tree decomposition $\cT$ of a graph $G$, define $\gamma(\cT)$ as the maximum value of $\gamma(B)$ taken over all subgraphs $B$ induced by a single bag of $\cT$.
By $\gamma_{\textrm{tw}}(G)$ we denote the minimum value of $\gamma(\cT)$ over all tree decompositions $\cT$ of $G$.\footnote{While not immediately relevant for our paper, it may be worth mentioning that such {tree decomposition} based parameters can also be defined and studied in the more general context of hypergraphs (see~\cite{adler2006width,DBLP:journals/jacm/Marx13,DBLP:conf/soda/Yolov18}).}
Note that if $\gamma(G) = |V(G)|-1$, then $\gamma_{\textrm{tw}}$ is precisely treewidth, and if $\gamma(G)$ is the size of a largest independent set in $G$, then $\gamma_{\textrm{tw}}$ is the tree-independence number.
Dallard et al.~\cite{dallard2022computing} showed that for every invariant $\gamma$ {as above}, either $\gamma_{\textrm{tw}}(G)$ is upper-bounded by a function of $\tin(G)$ (i.e., boundedness of $\gamma_{\textrm{tw}}(G)$ implies boundedness of $\tin(G)$), or \MWIS remains \NP-hard even {for} graphs $G$ with constant $\gamma_{\textrm{tw}}(G)$.

{Four years earlier}, Yolov~\cite{DBLP:conf/soda/Yolov18} (obviously unaware of the work of Dallard et al.~\cite{dallard2022computing}) had defined another, ``stronger'' parameter of a tree decomposition that can still be used to solve \MWIS. Of course there is no contradiction here -- Yolov's parameter does not fall into the category of ``natural'' parameters considered by Dallard et al.
In the original paper of Yolov~\cite{DBLP:conf/soda/Yolov18}, the parameter in question is called \emph{minor-matching hypertree width} and is defined in a much more general setting of hypergraphs.
As such a general definition is not relevant to our work, let us focus on the case of graphs; here we will call this parameter \emph{induced matching treewidth}.
For a tree decomposition $\cT$ of $G$, by $\mu(\cT)$ let us denote the size of a largest induced matching in $G$ {all of whose edges} intersect a single bag of $\cT$.
Now, the induced matching treewidth of $G$, denoted by $\yw(G)$, is the minimum value of $\mu(\cT)$ over all tree decompositions $\cT$ of $G$.
Note that $\mu(\cT)$ does not depend only on subgraphs induced by single bags, as some edges of an induced matching defining $\mu(\cT)$ might have one endpoint outside the bag.

It follows immediately from the definitions that for every graph it holds that $\yw(G) \leq \tin(G)$.
On the other hand, boundedness of induced matching treewidth does not imply boundedness of tree-independence number: for a biclique $K_{n,n}$ with each part of size $n$ we have $\yw(K_{n,n}) = 1$ but, as observed by Dallard et al.~\cite{dallard2022firstpaper}, $\tin(K_{n,n})=n$.
Thus indeed, induced matching treewidth is a ``stronger'' parameter.

{It is also important to note that the families of graph classes with bounded induced matching treewidth and classes of graphs with polynomially many minimal separators are incomparable.
For example, the class of all graphs $G_k$ consisting of $k$ internally disjoint paths of length three with the same endpoints has bounded treewidth and hence bounded induced matching treewidth, but graphs in the class have exponentially many minimal separators (it is not difficult to see that the graph $G_k$ has at least $2^k$ minimal separators).
On other other hand, $P_4$-free graphs have a polynomial number of minimal separators (see~\cite{MR2204116}) but unbounded induced matching treewidth.
This follows from a construction that we will present later in the paper, in \cref{prop:matching-biclique}.}

At first it is not clear why boundedness of induced matching treewidth is helpful in solving \MWIS.
However, the following structural result by Yolov~\cite{DBLP:conf/soda/Yolov18} is the key to understanding the connection between these two notions.
Let $\cT$ be a tree decomposition of $G$ with $\mu(\cT) \leq k$.
Then, for any node $t$ of $\cT$, in time $n^{\Oh(k)}$ we can enumerate a family $\misets_t$ such that
for any  inclusion-wise maximal independent set $I$ of $G$, the intersection of $I$ with the bag associated with $t$ is in $\misets_t$.
Since every \emph{maximum} independent set is in particular \emph{maximal}, this immediately implies that the number of ways an optimal solution might intersect each bag of $\cT$ is polynomial in $n$, if $k$ is a constant.
Using also the fact that given a graph with induced matching treewidth at most $k$, in polynomial time we can compute a tree decomposition such that {$\mu(\cT)=\Oh(k)$} (see \cref{thm:approximate-yolov} in \cref{sec:prelim}), we obtain the following algorithmic result as a consequence.

\begin{theorem}
\label{thm:yolovmwis}
For every fixed $k$, the \MWIS problem on $n$-vertex graphs $G$ with $\yw(G) \leq k$ can be solved in time {$n^{\Oh(k)}$}.
\end{theorem}

The tools developed by Yolov~\cite{DBLP:conf/soda/Yolov18} can also be used to solve other problems that boil down to finding a constant number of (maximal) independent sets, like $r$-\textsc{Coloring} or finding a homomorphism to a fixed target graph. Even though this was not stated explicitly in~\cite{DBLP:conf/soda/Yolov18}, one can also obtain a polynomial-time algorithm for the problem of finding a largest induced $r$-colorable subgraph. For $r=2$ this problem is equivalent (by taking the complement of the solution) to \textsc{Min Odd Cycle Transversal}.

\medskip
The main goal of this paper is to explore what other problems that do not fall into the above category can be solved in polynomial time in classes of bounded induced matching treewidth.
Recall that for bounded-treewidth graphs, a rich family of tractable problems is provided by the meta-theorem of Courcelle~\cite{DBLP:journals/iandc/Courcelle90}.
For graphs $G$ with polynomially many minimal separators, where the framework of potential maximal cliques can be applied, a somewhat similar general result is provided by Fomin, Todinca, and Villanger~\cite{FominTV15}. For fixed integer $r$ and fixed \cmsotwo formula $\psi$,
by $(r,\psi)$-\MWIS we denote the following computational problem (here ``\MWIS'' stands for \textsc{Max Weight Induced Subgraph}).
\problemTask{$(r,\psi)$-\MWIS}%
{A graph $G$ equipped with a weight function $\wei\colon V(G) \to \Q_+$.}%
{Find a set $F \subseteq V(G)$, such that
\begin{itemize}
\item $G[F] \models \psi$,
\item $\tw(G[F]) \leq r$,
\item $F$ is of maximum weight subject to the conditions above,
\end{itemize}
or conclude that no such set exists.}

Fomin, Todinca, and Villanger~\cite{FominTV15} proved that for each fixed $r$ and $\psi$, the $(r,\psi)$-\MWIS problem can be solved in time polynomial in the size of the input graph $G$ and the number of minimal separators in $G$.
Thus the running time is polynomial when restricted to classes of graphs with polynomial number of minimal separators.

Algorithms for $(r,\psi)$-\MWIS are also known for graphs excluding a fixed induced path, or graphs excluding long induced cycles~\cite{DBLP:conf/soda/AbrishamiCPRS21,DBLP:conf/stoc/GartlandLPPR21}.
We conjecture that $(r,\psi)$-\MWIS can also be solved in polynomial time for graphs of bounded induced matching treewidth.

\begin{conjecture}\label{conjecture}
For every fixed $k,r$ and a \cmsotwo formula $\psi$, the $(r,\psi)$-\MWIS{} {problem can be solved in polynomial time for graphs with induced matching treewidth at most $k$}.
\end{conjecture}

Even though we are not (yet) able to prove \cref{conjecture}, we provide some substantial evidence by approaching it from three different directions.

\section{Overview of our results}

The paper contains three main algorithmic contributions, as well as an initial set of results regarding induced matching treewidth, including computational results and bounds.
We give an overview of these four sets of results and their implications in the following subsections.

\subsection{Solving \textsc{Max Weight Induced Forest}.}
As our first and main result, we show that in graphs of bounded induced matching treewidth, one can in polynomial time find an induced forest of maximum possible weight.

\begin{restatable}{theorem}{thmfvs}
\label{thm:fvs}
For every fixed $k$, the \textsc{Max Weight Induced Forest} problem on $n$-vertex graphs $G$ with $\yw(G) \leq k$ can be solved in time {$n^{\Oh(k)}$}.
\end{restatable}

Note that \textsc{Max Weight Induced Forest} is equivalent to finding a maximum-weight induced subgraph of treewidth at most 1, i.e., is a special case of  $(r,\psi)$-\MWIS for $r=1$ and $\psi$ being any formula satisfied by all graphs.
Furthermore, by complementation, \textsc{Max Weight Induced Forest} is equivalent to the well-studied \textsc{Min Weight Feedback Vertex Set} problem.

Before sketching the proof of \cref{thm:fvs}, let us recall the textbook algorithm for \textsc{Max Weight Induced Forest} {for} graphs of bounded treewidth (see~\cite{platypus}).
Let $\cT$ be a tree decomposition of the instance graph, $t$ be a node of $\cT$, and $X_t$ be the bag corresponding to $t$.
For each set $Z \subseteq X_t$ and each partition $\pi$ of $Z$,
we keep the maximum weight of an induced forest $F$ contained in the subgraph of $G$ induced by the bags of the subtree of $\cT$ rooted at $t$, such that:
\begin{itemize}
\item the intersection of $V(F)$ with $X_t$ is exactly $Z$, and
\item the partition $\pi$ corresponds to the connected components of $F$.
\end{itemize}
For any induced forest $F$ in $G$, we call a pair $(Z,\pi)$ as above the \emph{signature of $F$ at $t$}.
Now, processing $\cT$ in a bottom-up fashion, we can find a maximum-weight induced forest in $G$ using the information stored for each node.

The key insight leading to the proof of \cref{thm:fvs} is the following structural lemma.

\begin{lemma}[Simplified statement of \cref{lem:ff}]
\label{lem:ff-inf}
Let $k$ be a fixed integer. For an $n$-vertex graph $G$, a tree decomposition $\cT$ of $G$ with $\mu(\cT) \leq k$, and a node $t$ of $\cT$,
in time $n^{\Oh(k)}$ we can enumerate a set $\FF_t$ that contains the signature at $t$ of every maximal induced forest in $G$.
\end{lemma}

Recall that an analogous result for \MWIS was shown by Yolov~\cite{DBLP:conf/soda/Yolov18}: for each bag $X_t$ corresponding to a node $t$ of a tree decomposition $\cT$ of bounded $\mu(\cT)$, in polynomial time we can enumerate a family $\misets_t$ that contains an intersection of each maximal independent set in $G$ with~$X_t$.

As every maximum induced forest is in particular  inclusion-wise maximal,  \cref{lem:ff-inf} allows us to solve \textsc{Max Weight Induced Forest} by essentially mimicking the textbook algorithm for bounded-treewidth graphs.
We just need to find the set $\FF_t$ for each node $t$ of $\cT$, and instead of indexing the dynamic programming table by \emph{all} possible pairs $(Z,\pi)$, we use just the pairs that are in $\FF_t$.

So let us describe the proof of \cref{lem:ff-inf}.
Let $F$ be a maximal induced forest in $G$; we think of it as an (unknown) optimal solution. Consider a node $t$ of $\cT$ and its corresponding bag $X_t$.
The difficulty that we need to face is that the intersection of $F$ with $X_t$ might be arbitrarily large; for example, a component of $F$ might be a large induced star contained in $X_t$. Thus we need to find some compact way to encode such possible intersections.

Let $\skel{F}$ be the \emph{skeleton of $F$}, i.e., the set of vertices of degree at least two in $F$ {together with an arbitrary but fixed vertex $v_C$ of each two-vertex component $C$ of $F$}.
By $L(F)$ and $T(F)$ we denote, respectively, the sets of vertices of degree 1 (\emph{leaves}) and of degree 0 (\emph{trivial} vertices) of $F$, with the exception that for each two-vertex component $C$ of $F$ we include {in $L(F)$ exactly one vertex from $C$, namely the vertex of $C$ different from $v_C$.}
Clearly, the sets $\skel{F}, L(F), T(F)$ form a partition of the vertex set of $F$.

The first important observation is that the set $S \coloneqq \skel{F} \cap X_t$ is of size bounded by a function of $k$.
Indeed, if $S$ was too large, we could easily extract from $F$ an induced matching whose every edge intersects the bag.
Thus we can directly enumerate all candidates for $\skel{F} \cap X_t$.

So we are left with finding an encoding of the set $(L(F) \cup T(F)) \cap X_t$.
Observe that $\limb{F} \coloneqq L(F) \cup T(F)$ is an independent set in $G$.
However, this is not very useful, as $L(F) \cup T(F)$ does not have to be maximal and we only have some information about the intersections of maximal independent sets with $X_t$ (they have to be in $\misets_t$, which is of polynomial size).
Let $I^*(F)$ be some maximal independent set of $G$ containing $\limb{F}$; we know that its intersection $I$ with $X_t$ can only be chosen in a polynomial number of ways.
However, $I$ contains some vertices that are not in the solution (called \emph{impostors}) and we should not consider their weight when choosing an optimal solution.
Why does a vertex $v \in I$ not belong to $\limb{F}$?
As $F$ is a maximal forest, there are two possible reasons why $v$ is an impostor:
\begin{itemize}
\item $v \in \skel{F}$, or
\item adding $v$ to $F$ would create a cycle.
\end{itemize}
Note that the vertices of the first type can be easily recognized as $S=\skel{F} \cap X_t$ is directly represented.
On the other hand, each vertex of the second type is adjacent to at least two vertices from some connected component of $F$.
Furthermore, we know that these two vertices are in $\skel{F}$, as $\limb{F} \cup \{v\} \subseteq I^*(F)$ and $I^*(F)$ is independent.
If these two neighbors of $v$ in $\skel{F}$ are in $X_t$, we are done, as we can guess the set $S=\skel{F} \cap X_t$.
However, the neighbors may lie outside the bag $X_t$.

Summing up, the type of a vertex $v \in I \setminus S$ (which may or may not belong to $\limb{F}$) is determined by its neighborhood in $\widehat{S} \coloneqq \skel{F} \cap N[X_t]$, where by $N[X_t]$ we denote the set consisting of the vertices from $X_t$ and their neighbors.
If $v$ has no neighbors in $\widehat{S}$, then $v$ is trivial, i.e., $v \in T(F)$.
If $v$ has one neighbor in $\widehat{S}$, then $v$ is a leaf, i.e., $v \in L(F)$.
Finally, if $v$ has at least two neighbors in $\widehat{S}$, then $v$ is an impostor, i.e., $v \in  I \setminus \limb{F}$.
If we could store the information about the set $\widehat{S}$ explicitly, we could easily distinguish between these three types of vertices.
Unfortunately, the set $\widehat{S}$ might still be arbitrarily large.
So we need to make one more step: we observe that an  inclusion-wise minimal subset $Q$ of $\widehat{S}$ that still can distinguish between the three types of vertices in $I \setminus S$ is of size bounded by a function of $k$, as otherwise we could again extract from $F$ a large induced matching whose every edge intersects $X_t$.
Thus we can guess every candidate for~$Q$.

From the argument above it follows that the triple $(S,I,Q)$ can be exhaustively guessed in a polynomial number of ways. Furthermore, thanks to the properties of $Q$, we can uniquely extract $L(F) \cap X_t$ and $T(F) \cap X_t$  from $I$.
This way we can compute $Z$, i.e., the intersection of $F$ with $X_t$.

In order to obtain $\pi$, we can exhaustively guess the partition $\widehat\pi$ of $S \cup Q$ that corresponds to the components of $F$ restricted to the subgraph of $G$ induced by the bags in the subtree rooted at $t$. Here we use the fact that the size of $S \cup Q$ is bounded by a function of $k$, i.e., a constant.
From $\widehat\pi$ we can uniquely reconstruct $\pi$, as each vertex in $L(F) \cap X_t$ has a neighbor in $S \cup Q$, and each vertex in $T(F)$ forms a separate component.
This completes the sketch of proof of \cref{lem:ff-inf}.

\subsection{Independent packings of small subgraphs.}
Let $\cH = \{H_j\}_{j \in J}$ be a family of connected subgraphs of a graph $G$.
By $G^\circ[\cH]$ we denote the graph with vertex set $J$, where two distinct vertices $j,j'\in J$ are adjacent if $H_j$ and $H_{j'}$ have a vertex in common or there is an edge with one endpoint in $H_j$ and the other in $H_{j'}$.\footnote{This extends the notation $G^\circ$ used in~\cite{DBLP:conf/stoc/GartlandLPPR21} for the graph  $G^\circ[\cH]$ in the particular case when $\cH$ is the set of all non-null connected subgraphs of $G$.}
Such a construction was considered by Cameron and Hell in~\cite{MR2190818}, who focused on the particular case when $\cH$ is the set of all subgraphs of $G$ isomorphic to a member of a fixed family $\mathcal{F}$ of connected graphs; they showed that if $G$ is a chordal graph, then so is $G^\circ[\cH]$.
Dallard et al.~\cite{dallard2022firstpaper} generalized this result by showing that $\tin(G^\circ[\cH]) \leq \tin(G)$.
Other graph classes closed under taking $G^\circ[\cH]$ are classes of graphs excluding long induced paths or long induced cycles; see Gartland et al.~\cite{gartland2020finding}.
We show that an analogous statement holds for graphs with bounded induced matching treewidth.

\begin{lemma}[A part of the statement of \cref{lem:H-G}]
\label{lem:H-G-simplified}
Let $G$ be a graph and let $\mathcal{H}$ be a set of connected non-null subgraphs of $G$.
Then $\yw(G^\circ[\cH])\le \yw(G)$.
\end{lemma}

\cref{lem:H-G-simplified} yields several algorithmic corollaries.
First, notice that the \MWIS problem on $G^\circ[\cH]$ (with some weight function defined on $J$) corresponds to the problem of packing pairwise disjoint and nonadjacent graphs from $\cH$, called  \textsc{Max Weight Independent Packing}. Thus, combining \cref{lem:H-G} with \cref{thm:yolovmwis},
we show that \textsc{Max Weight Independent Packing} in classes of bounded induced matching treewidth can be solved in time polynomial in the size of $G$ and $\cH$.
When $\cH$ is the set of all subgraphs of $G$ isomorphic to a member of a fixed family $\mathcal{F}$ of connected graphs, then \textsc{Max Weight Independent Packing} coincides with the \textsc{Max Weight Independent $\mathcal{F}$-Packing} problem studied in Dallard et al.~\cite{dallard2022firstpaper}.
This latter problem in turn generalizes several problems studied in the literature:
\begin{itemize}
    \item the \MWIS problem, which corresponds to the case $\mathcal{F} = \{K_1\}$,
    \item the \textsc{Max Weight Induced Matching} problem (see, e.g.,~\cite{MR3776983,MR4151749}), which corresponds to the case $\mathcal{F} = \{K_2\}$,
    \item the \textsc{Dissociation Set} problem (see, e.g.,~\cite{MR3593941,MR615221,MR2812599}), which corresponds to the case when $\mathcal{F}= \{K_1,K_2\}$ and the weight function assigns to each subgraph $H_j$ the weight equal to $|V(H_j)|$,
    \item the \textsc{$k$-Separator} problem (see, e.g.,~\cite{DBLP:conf/soda/Lee17,MR3987192,MR3296270}), which corresponds to the case when $\mathcal{F}$ contains all connected graphs with at most $k$ vertices, the graph $G$ is equipped with a vertex weight function $\wei:V(G)\to \mathbb{Q}_+$, and the weight function on $\cH$ assigns to each subgraph $H_j$ the weight equal to $\sum_{x\in V(H_j)}\wei(x)$.
\end{itemize}

Observe that \textsc{Max Weight Induced Forest} for a weighted graph $G$ is equivalent to \textsc{Max Weight Independent Packing} for an instance $G$, $\cH$ (and $\wei$), where $\cH$ consists of all induced subtrees of $G$ (and the weight of $H_j$ is the total weight of the vertices in $H_j$).
{Despite that}, we note that the combination of \cref{lem:H-G-simplified} and \cref{thm:yolovmwis} cannot be used to solve \textsc{Max Weight Induced Forest} in polynomial time. Indeed, the size of $\cH$ is not bounded by a polynomial function of $|V(G)|$.

However, \cref{lem:H-G-simplified} might still be of use in this context.
As already observed by Gartland et al.~\cite{gartland2020finding}, we can sacrifice a small part of the optimal solution to break it into constant-size parts, which can then be handled efficiently.
Indeed, for each $\epsilon >0$ and any forest $F$ we can remove an $\epsilon$-fraction of the vertices of $F$ so that the remaining components are of constant size (where this constant depends on $\epsilon$ but not on $F$).
This approach can be used to obtain a PTAS for \textsc{Max Induced Forest} (the unweighted variant of \textsc{Max Weight Induced Forest}).
Actually, this reasoning goes far beyond trees and can be used to pack large induced subgraphs from any \emph{weakly hyperfinite} class; see~\cref{sec:blob} for the definition.
In particular, we obtain the following result which provides further evidence for \cref{conjecture}.

\begin{restatable}{corollary}{ptas}
\label{cor:ptas}
For every fixed $k,r \in \N$ and $\epsilon >0$, given a graph $G$ with induced matching treewidth at most $k$,
in polynomial time we can find a set $F \subseteq V(G)$ such that:
\begin{enumerate}
\item $\tw(G[F]) \leq r$,
\item the size of $F$ is at least $(1-\epsilon) \textsf{OPT}$, where  $\textsf{OPT}$ is the size of a largest set satisfying the first condition.
\end{enumerate}
\end{restatable}

Note that the problem addressed in \cref{cor:ptas} is a special case of
 $(r,\psi)$-\MWIS{} {obtained by taking $\psi$ to be} any formula satisfied by all graphs and {$\wei$ to be a uniform weight function.}

\begin{sloppypar}
Next, we generalize the polynomial-time solvability of the \textsc{Max Weight Independent Packing} problem {for} graphs of bounded induced matching treewidth to \textsc{Max Weight Distance-$d$ Packing}: the problem of packing subgraphs at distance $d$, for all even positive integers $d$.
The case $d = 2$ is precisely \textsc{Max Weight Independent Packing}.
We remark that unless $\P = \NP$, this result cannot be generalized to odd values of $d$, since (as shown by Eto, Guo, and Miyano~\cite{DBLP:journals/jco/EtoGM14}), the distance-$3$ variant of \textsc{Max Independent Set} problem is $\NP$-hard for chordal graphs, which have induced matching treewidth (and even tree-independence number) at most one.
\end{sloppypar}

This extension again follows from a structural observation that may be of interest on its own.
Given a graph $G$ and a positive integer $k$, we denote by $G^k$ the \emph{$k$-th power} of $G$, that is, the graph obtained from the graph $G$ by adding to it all edges between pairs of {distinct} vertices at distance at most $k$.
We show that for any graph $G$ and any positive integer $k$, the induced matching treewidth of the graph $G^{k+2}$ cannot exceed that of the graph $G^k$, and the same relationship holds for the tree-independence number.
This is a significant generalization of a result of Duchet~\cite{MR778751}, who proved an analogous result for graphs with tree-independence number at most one: if $G^k$ is chordal, then so is $G^{k+2}$.
As a consequence, we obtain that for every positive integer $k$, the class of graphs with induced matching treewidth (resp.~tree-independence number) at most $k$ is closed under taking odd powers, which generalizes the analogous result for the class of chordal graphs proved by Balakrishnan and Paulraja~\cite{MR704427}.
We complement these results by observing that the class of even powers of chordal graphs (for any fixed power) is not contained in any nontrivial hereditary graph class.
This implies in particular that any such class has unbounded tree-independence number and induced matching treewidth.

\subsection{\boldmath Solving $(r,\psi)$-MWIS for bounded tree-independence number.}
Finally, we show that $(r,\psi)$-\MWIS can be solved in polynomial time for graphs of bounded tree-independence number.
As shown by Dallard et al.~\cite{dallard2022secondpaper}, many natural graph classes fall into this category. Furthermore, they all have bounded induced matching treewidth.

Actually, we show tractability of a more general problem, where instead of asking for a subgraph of bounded treewidth, we ask for a subgraph of clique number bounded by $r$; let us call this variant $(\omega \leq r,\psi)$-\MWIS.
We remark that $(\omega \leq r,\psi)$-\MWIS is a generalization of $(r-1,\psi)$-\MWIS, as every graph of treewidth at most $r-1$ has clique number at most $r$ and the property of being of bounded treewidth can be expressed in \cmsotwo~\cite[Lemma 10]{gartland2020finding}.
On the other hand, there are natural classes of graphs of bounded clique number and unbounded treewidth, e.g., bipartite graphs or planar graphs.
However, one can observe that {for} graphs of bounded tree-independence number, treewidth is upper-bounded by a function of the clique number {(see~\cite{dallard2022firstpaper}); this phenomenon is called \emph{$(\tw,\omega)$-boundedness} in the literature (see, e.g.,~\cite{dallard2021treewidth}).}
Thus we conclude that in classes of bounded tree-independence number, both $(\omega \leq r,\psi)$-\MWIS and $(r,\psi)$-\MWIS formalisms describe the same family of problems.

\begin{theorem}[Simplified statement of \cref{thm:cmsomain-intro}]
\label{thm:cmso-simple}
For every fixed $k,r$ and a \cmsotwo formula $\psi$,
the $(\omega \leq r,\psi)$-\MWIS{} {for} graphs with tree-independence number at most $k$ can be solved in polynomial time.
\end{theorem}

Some examples of problems that can be solved in polynomial time with this approach include:
\begin{itemize}
    \item finding a largest induced planar subgraph (which is equivalent to \textsc{Planarization}~\cite{DBLP:conf/soda/JansenLS14,DBLP:journals/dam/Pilipczuk17}),
    \item finding a largest induced odd cactus (which is equivalent to \textsc{Even Cycle Transversal}~\cite{DBLP:journals/siamdm/PaesaniPR22,DBLP:conf/wg/MisraRRS12,DBLP:conf/iwpec/BergougnouxBBK20}),
    \item finding the maximum number of pairwise disjoint and non-adjacent cycles (for this problem we need a slightly stronger variant of \cref{thm:cmso-simple} that is shown in \cref{sec:cmso}).
\end{itemize}

We point out that the bound on the clique number of the sought-for subgraph must be constant.
Indeed, the property of being a clique is easily expressible in \cmsotwo, but \textsc{Max Clique} is \NP-hard {for} graphs with tree-independence number 2.

\begin{sloppypar}
\subsection{Complexity of computing induced matching treewidth and bounds}\label{subsec:structural}
\end{sloppypar}

\begin{sloppypar}
We observe that known relations between the induced matching treewidth the tree-independence number~\cite{DBLP:conf/soda/Yolov18} together with known results on tree-independence number~\cite{dallard2022firstpaper,dallard2022computing} and  independence number~\cite{MR2403018} imply several hardness results related to exact and approximate computation of the induced matching treewidth.
In particular, for every constant $k\geq 4$, it is \NP-complete to determine whether $\yw(G)\leq k$ for a graph~$G$ (see \Cref{thm:hardness} for the precise statement of the results).
While the complexity of recognizing graphs with induced matching treewidth $k$ for $k\in \{2,3\}$ remains open, we then show that the problem is solvable in polynomial time for $k = 1$.
To this end, we first prove that induced matching treewidth is monotone under induced minors, that is, it cannot decrease upon deleting a vertex or contracting an edge.
Then, we use a characterization of graphs $G$ such that the square of the line graph of $G$ is chordal (see~\cite{MR3804753}) to characterize the class of graphs with induced matching treewidth at most $1$ in terms of a finite family of forbidden induced minors.
This characterization implies that the family of forbidden induced \emph{subgraphs} for this class is very restricted, containing only finitely many graphs except for cycles of length at least $6$, which immediately leads to a polynomial-time recognition algorithm.
\end{sloppypar}

\medskip
We then consider the behavior of induced matching treewidth for graphs with bounded degree, in particular, its relation to tree-independence number and treewidth.
While any graph class with unbounded induced matching treewidth is necessary of unbounded tree-independence number and hence of unbounded treewidth, it is known that the converse implications hold in the absence of a fixed complete bipartite graph as a subgraph (see p.~\pageref{thm:tw-biclique} for details).
In particular, for any class of graphs with bounded maximum degree, all the aforementioned parameters (induced matching treewidth, tree-independence number, and treewidth) are all equivalent to each other, in the sense that they are either all bounded or all unbounded.
We strengthen this result by showing that for any class of graphs with bounded maximum degree, the three parameters are in fact linearly related to each other (see~\cref{thm:delta,cor:treewidth}).

\medskip
In conclusion, we consider two families of graphs with unbounded maximum degree.
We give a lower bound on the induced matching treewidth of hypercube graphs and identify a family of $P_4$-free graphs with unbounded induced matching treewidth.

\bigskip
\paragraph*{Organization of the paper}
In \cref{sec:prelim}, we provide the necessary preliminaries.
In \cref{sec:fvs}, we prove \cref{thm:fvs}.
In \cref{sec:blob}, we show \cref{lem:H-G} and discuss its algorithmic consequences, in particular, \cref{cor:ptas}.
In \cref{sec:cmso}, we show \cref{thm:cmso-simple} and its strengthened variant.
In \cref{sec:structural} we initiate a systematic study of induced matching treewidth, by showing several computational results and bounds.
The paper is concluded with pointing out some directions for further research in \cref{sec:conclusion}.

\section{Preliminaries} \label{sec:prelim}
For an integer $n$, by $[n]$ we denote the set $\{1,\ldots,n\}$.
Let $G$ be a graph.
For a set $X \subseteq V(G)$, by $G[X]$ we denote the subgraph of $G$ induced by $X$.
If the graph $G$ is clear from the context, we will identify induced subgraphs with their vertex sets.

For a vertex $v$ of $G$, by $N(v)$ we denote the set of neighbors of $v$.
For a set $X \subseteq V(G)$, by $N(X)$ we denote the set $\left(\bigcup_{v \in X} N(v)\right) \setminus X$, and by $N[X]$ we denote the set $N(X) \cup X$.
{For a vertex $v$ of $G$, we write $N[v]$ for $N[\{v\}]$.}

Let $X,Y$ be disjoint subsets of $V(G)$.
We say that they are \emph{non-adjacent} if there is no edge with one endpoint in $X$ and the other in $Y$.
For $k\in \{1,2\}$, the set $Y$ \emph{$k$-dominates} $X$ if each vertex from $X$ has at least $k$ neighbors in $Y$.

A \emph{clique} in a graph $G$ is a set of pairwise adjacent vertices; an \emph{independent set} is a set of pairwise non-adjacent vertices.
The \emph{clique number} and the \emph{independence number} of $G$, denoted respectively by $\omega(G)$ and $\alpha(G)$, are defined as the maximum cardinality of a clique, resp.~independent set, in~$G$.

A \emph{matching} in a graph $G$ is a set $M$ of pairwise disjoint edges.
If, in addition, no two vertices from different edges in the matching $M$ are adjacent, then the matching $M$ is an \emph{induced matching} in $G$.
An edge \emph{touches} a set $X \subseteq V(G)$ if it has at least one endpoint in $X$.
A set $M$ of edges \emph{touches} $X$ if every edge from $M$ touches $X$.

The \emph{distance} between two vertices $u$ and $v$ in $G$ is denoted by $\dist_G(u,v)$ and defined as the length of a shortest path between $u$ and $v$ (or $\infty$ if there is no $u$-$v$-path in $G$).
Given a positive integer $k$, the \emph{$k$-th power} of $G$ is the graph $G^k$ with vertex set $V(G)$, in which two distinct vertices $u$ and $v$ are adjacent if and only if $\dist_G(u,v)\le k$. Note that $G^1 =G$.

Let {$\wei\colon V(G) \to \Q_+$} be a weight function.
For a set $X \subseteq V(G)$, by $\wei(X)$ we denote $\sum_{v \in X} \wei(v)$.
In particular, $\wei(\emptyset) = 0$.

\paragraph{Tree decompositions.}
A \emph{tree decomposition} of a graph $G$ is a pair $\mathcal{T} = (T, \{X_t\}_{t\in V(T)})$ where $T$ is a tree and every node $t$ of $T$ is assigned a vertex subset $X_t\subseteq V(G)$ called a \emph{bag} such that the following conditions are satisfied: every vertex is in at least one bag, for every edge $uv\in E(G)$ there exists a node $t\in V(T)$ such that $X_t$ contains both $u$ and $v$, and for every vertex $u\in V(G)$ the subgraph of $T$ induced by the set $\{t\in V(T): u\in X_t\}$ is connected (that is, a tree).

Usually we consider tree decompositions to be rooted, and the root is denoted by $\root(T)$.
For a node $t$ of $T$, by $T_t$ we denote the subtree of $T$ rooted at $t$, and by $G_t$ we denote the subgraph of $G$ induced by $V_t \coloneqq \bigcup_{t' \in V(T_t)} X_{t'}$.
A rooted tree decomposition $\cT = (T,\{X_t\}_{t \in V(T)})$ of a graph $G$ is \emph{nice}
if each node $t$ is of one of the following types:
\begin{description}
\item[leaf] a {leaf in $T$} such that $X_t=\emptyset$,
\item[introduce] an inner node with a single child $t'$ such that $X_t = X_{t'} \cup \{v\}$ for some $v \in V(G) \setminus X_{t'}$,
\item[forget] an inner node with a single child $t'$ such that $X_t = X_{t'} \setminus \{v\}$ for some $v \in X_{t'}$,
\item[join] an inner node with exactly two children $t',t''$, such that $X_t = X_{t'} = X_{t''}$,
\end{description}
and $X_{\root(T)} = \emptyset$.
It is well known that given a tree decomposition $\cT=(T,\{X_t\}_{t \in V(T)})$ of $G$, one can obtain in time $\Oh(|V(G)|^2|V(T)|)$ a nice tree decomposition $\cT'=(T',\{X'_t\}_{t \in V(T')})$ of $G$ with at most $|V(G)| \cdot |V(T)|$ nodes and whose every bag is contained in some bag of $\cT$ (see, e.g., the textbook~\cite{platypus}, or~\cite{dallard2022firstpaper} for a treatment focused on the independence number).

\paragraph{Tree-independence number and induced matching treewidth.}
Consider a graph $G$ and a tree decomposition $\mathcal{T} = (T, \{X_t\}_{t\in V(T)})$ of $G$.
The \emph{independence number} of $\mathcal T$, denoted by $\alpha_G(\mathcal{T})$, is defined as follows:
\[\alpha_G(\mathcal{T}) = \max_{t\in V(T)} \alpha(G[X_t])\,.\]

The \emph{tree-independence number} of $G$, denoted by $\tin(G)$, is the minimum independence number among all possible tree decompositions of $G$.
A graph $G$ is chordal if and only if $\tin(G)\le 1$ (see~\cite{dallard2022firstpaper}).

Let again $G$ be a graph and $\mathcal{T} = (T, \{X_t\}_{t\in V(T)})$ be its tree decomposition.
Let
\[\mu_G(\mathcal{T}) = \max\left\{|M|\colon \textnormal{ $M$ is an induced matching in $G$ that touches a bag of } \mathcal{T}\right\}\,.\]
The \emph{induced matching treewidth} of $G$, denoted by $\yw(G)$, is the minimum value of $\mu_G(\cT)$ among all possible tree decompositions $\cT$ of $G$. Recall that every graph $G$ satisfies $\yw(G)\le \tin(G)$.

Note that if $\cT$ and $\cT'$ are tree decompositions such that every bag of $\cT'$ is contained in some bag of $\cT$ (in particular, if $\cT'$ is a nice tree decomposition obtained from $\cT$), then we have $\alpha_G(\cT') \leq \alpha_G(\cT)$ and $\mu_G(\cT') \leq \mu_G(\cT)$.
If the graph $G$ is clear from the context, we simply write $\alpha(\cT)$ and $\mu(\cT)$ for $\alpha_G(\cT)$ and $\mu_G(\cT)$, respectively.

{We now recall two results by Yolov from~\cite{DBLP:conf/soda/Yolov18}, expressing each of these two parameters in terms of the value of the other parameter on a derived graph.
Let $G=(V,E)$ be a graph.
We denote by $L^2(G)$ the \emph{square of the line graph} of $G$, that is, the graph with vertex set $E$ in which two distinct vertices $e,f\in E$ are adjacent if and only if the subgraph of $G$ induced by the endpoints of $e$ and $f$ is connected.
Furthermore, by $G\odot K_1$ we denote the \emph{corona} of $G$, that is, the graph obtained from $G$ by adding a pendant edge to each vertex.
Formally, $V(G\odot K_1) = V\cup V'$
where $V' = \{v':v\in V\}$ is a set of $|V|$ new vertices, and $E(G\odot K_1) = E\cup \{vv':v\in V\}$.

\begin{proposition}[Yolov~\cite{DBLP:conf/soda/Yolov18}]
\label{lem:Yolov}
For every graph $G$, it holds that \[\yw(G) = \tin(L^2(G))\quad\quad \textrm{and} \quad\quad\tin(G) = \yw(G\odot K_1)\,.\]
\end{proposition}
}

Furthermore, given a graph with bounded {tree-independence number or bounded} induced matching treewidth, we can compute in polynomial time a tree decomposition with bounded {$\alpha(\cT)$ or $\mu(\cT)$, respectively.
Dallard et al.~showed in~\cite{dallard2022computing} that there is an algorithm that, given an $n$-vertex graph $G$ and an integer $k$, in time $2^{\Oh(k^2)}n^{\Oh(k)}$ either outputs a tree decomposition of $G$ with with independence number at most $8k$, or determines that the tree-independence number of $G$ is larger than~$k$.
This improves on an earlier analogous result of Yolov~\cite{DBLP:conf/soda/Yolov18} giving an algorithm with running time $n^{\mathcal{O}(k^3)}$ and a bound on the independence number of $\mathcal{O}(k^3)$.
A~consequence of the aforementioned result is the following.

\begin{theorem}[Dallard et al.~\cite{dallard2022computing}]\label{thm:approximate-dallard}
Let $k$ be a positive integer and let $G$ be an $n$-vertex graph satisfying $\tin(G) \leq k$.
Then, in time $n^{\Oh(k)}$ we can obtain a tree decomposition $\cT$ of $G$ such that $\alpha(\cT)\le 8k$.
\end{theorem}

Combining \cref{thm:approximate-dallard} with \cref{lem:Yolov} (and its constructive proof) yields a similar result for the induced matching treewidth.

\begin{theorem}
\label{thm:approximate-yolov}
Let $k$ be a positive integer and let $G$ be an $n$-vertex graph satisfying \hbox{$\yw(G) \leq k$}.
Then, in time $n^{\Oh(k)}$ we can obtain a tree decomposition $\cT$ of $G$ such that $\mu(\cT)\le 8k$.
\end{theorem}

\begin{proof}
By \cref{lem:Yolov}, $\tin(L^2(G)) = \yw(G) \le k$.
Note that the graph $L^2(G)$ has $ \mathcal{O}(n^2)$ vertices and can be computed in time $n^{\Oh(1)}$.
Applying \cref{thm:approximate-dallard} to the graph $L^2(G)$, we obtain in time $n^{\Oh(k)}$ a tree decomposition $\cT'$ of $L^2(G)$ such that $\alpha(\cT')\le 8k$.
Following the proof of~\cite[Theorem 4.2]{DBLP:conf/soda/Yolov18} (establishing $\yw(G)=\tin(L^2(G))$), this tree decomposition can be turned in polynomial time into a  tree decomposition $\cT$ of $G$ such that $\mu(\cT)\le 8k$.
\end{proof}
}

The following property of tree decompositions with bounded induced matching treewidth allows us to exploit them algorithmically, in order to solve the \textsc{MWIS} problem.

\begin{sloppypar}
\begin{theorem}[Yolov~\cite{DBLP:conf/soda/Yolov18}]\label{prop:yolov-misets}
Let $k$ be a positive integer, let $G$ be an $n$-vertex graph, and let $\cT=(T,\{X_t\}_{t \in V(T)})$ be a tree decomposition of $G$ such that $\mu(\cT) \leq k$.
Then for every $t \in V(T)$, in time $n^{\Oh(k)}$ one can enumerate a family $\misets_t$ of at most $n^{2k+2}$ subsets of $X_t$ with the following property:
For any maximal independent set $I$ of $G$, the set $I \cap X_t$ belongs to $\misets_t$.
\end{theorem}
\end{sloppypar}

We conclude this section with a short alternative proof of \cref{prop:yolov-misets}.
We include it for two reasons.
First, we believe it is easier to follow (at least in the context of the current paper), as the terminology and notation of~\cite{DBLP:conf/soda/Yolov18} is quite heavy and very different from ours.
Second, it serves as a warm-up for the main combinatorial result of \cref{sec:fvs}, i.e., \cref{lem:ff}, which uses the same basic idea, but is significantly more complicated.
We remark that the presented argument gives a slightly worse bound on the size of $\misets_t$; we do not attempt to optimize it, rather aiming for simplicity.

\begin{proof}[Proof of \cref{prop:yolov-misets} (with the bound of $n^{3k}$ on the size of $\misets_t$).]
Fix some $t \in V(T)$.
Define
\[\misets^* = \{ J \colon \exists \text{ a maximal independent set $I$ in $G$ s.t.~$I \cap X_t =J$}\}.\]
We will show that in time $n^{\Oh(k)}$ we can enumerate a set $\misets_t$ of size at most $n^{3k}$ that contains $\misets^*$.

Fix some $J \in \misets^*$ and let $I$ be a maximal independent set in $G$ such that $I \cap X_t = J$.
There exists some set $J'$ that is a maximal independent set in $G[X_t]$ and contains $J$.
Notice that every vertex of $J' \setminus J$ is adjacent to some vertex in $I \cap N(X_t)$. Let $Q$ be an inclusion-wise minimal subset of $I \cap N(X_t)$ such that every vertex from $J' \setminus J$ has a neighbor in $Q$.

Notice that by the minimality of $Q$ each element $q \in Q$ has a \emph{private neighbor} $r_q \in J' \setminus J$, i.e., $qr_{q'}$ is an edge if and only if $q = q'$. Furthermore, as each of $Q$ and $J' \setminus J$ is independent, we observe that
the vertices $\bigcup_{q \in Q}\{q,r_q\}$ induce in $G$ a matching of size $|Q|$ that touches $X_t$.
Consequently, $|Q| \leq k$.
Now observe that the pair $(J', Q)$ uniquely determines $J$ as $J=J' \setminus N(Q)$.

By a result of Alekseev~\cite[Corollary 1]{Alekseev+2007+355+359}, if a graph $H$ has no induced matching of size larger than $k$, then the number of maximal independent sets in $H$ is at most $|V(H)|^{2k}$. As $G[X_t]$ is such a graph, we conclude that the number of candidates for $J'$ is at most $n^{2k}$.
Furthermore, maximal independent sets can be enumerated with polynomial delay~\cite{DBLP:journals/siamcomp/TsukiyamaIAS77}, so we can enumerate all candidates for $J'$ in time $n^{\Oh(k)}$.
On the other hand, there are clearly at most $n^k$ candidates for $Q$ and they can be enumerated in time $n^{\Oh(k)}$. This completes the proof.
\end{proof}

\section{Solving \textsc{Max Weight Induced Forest}}\label{sec:fvs}
Let $F$ be an induced forest in $G$ and $\cT = (T, \{X_t\}_{t \in V(T)})$ be a tree decomposition of $G$.
The \emph{signature of $F$ at $t \in V(T)$} is a pair $(Z,\pi)$, such that:
\begin{enumerate}
    \item $Z = F \cap X_t$, and
    \item $\pi$ is the partition of $Z$ such that $z_1,z_2 \in Z$ are in the same block of $\pi$ if and only if there is a $z_1$-$z_2$-path in $F \cap V_t$.
\end{enumerate}
In other words, $\pi$ corresponds to connected components of the subgraph of $F$ induced by $V_t$.

Note that, given a tree decomposition $\cT$ with $\mu(\cT) \leq k$, we cannot bound the number of possible signatures of any induced forest at each bag of $\cT$ by a function of $k$. Indeed, it might happen that a bag induces an independent set of size linear in $n = |V(G)|$; then, every subset of such a bag forms an induced forest in $G$. However, this is no longer the case if we talk about \emph{ inclusion-wise maximal} induced forests; note that we can safely assume that every optimal solution is maximal.
The main combinatorial insight used in the proof of \cref{thm:fvs} is the following lemma, which is an analogue of \cref{prop:yolov-misets} for the \textsc{Max Weight Induced Forest} problem.

\begin{lemma}
\label{lem:ff}
Let $k$ be a fixed integer.
Let $G$ be an $n$-vertex graph and let $\cT=(T,\{X_t\}_{t \in V(T)})$ be a tree decomposition of $G$ satisfying $\mu(\cT) \leq k$.
Then, for each $t \in V(T)$ there is a set $\FF_t$ of size at most $(12k)^{12k} \cdot n^{14k+2}$ with the following property:
For every maximal induced forest $F$ of $G$, the signature of $F$ at $t$ is in $\FF_t$.
Furthermore, $\FF_t$ can be enumerated in time $n^{\Oh(k)}$.
\end{lemma}

Equipped with \cref{lem:ff}, we can basically mimic the textbook algorithm for \textsc{Max Weight Induced Forest}~\cite{platypus}.
This is encapsulated in the following lemma.

\begin{lemma}\label{lem:algo-fvs}
Let $k$ be a fixed integer.
Let $G$ be an $n$-vertex graph equipped with a vertex weight function $\wei$ and a tree decomposition $\cT=(T,\{X_t\}_{t \in V(T)})$ satisfying $\mu(\cT) \leq k$.
Assume that for each $t \in V(T)$ we are given a set $\FF_t$ as in \cref{lem:ff}.
Then, we can solve the \textsc{Max Weight Induced Forest} problem in time $n^{\Oh(k)} \cdot |V(T)|^{\Oh(1)}$.
\end{lemma}

Let us postpone the proofs of \cref{lem:ff,lem:algo-fvs} and now let us show that they indeed imply \cref{thm:fvs}.
Let $G$ be a graph satisfying $\yw(G) \leq k$, equipped with a weight function $\wei\colon V(G) \to \Q_+$.
We call the algorithm from~\cref{thm:approximate-yolov} to obtain a tree decomposition $\cT'$ with {$\mu(\cT')  =\Oh(k)$}.
The running time of this algorithm is {$n^{\Oh(k)}$.}
In polynomial time we can transform $\cT'$ into a nice tree decomposition $\cT = (T, \{X_t\}_{t \in V(T)})$ with $\mu(\cT) \leq \mu(\cT')$.
Define $k' \coloneqq \mu(\cT)$.

Next, for each $t \in V(T)$ we use \cref{lem:ff} to compute the set $\FF_t$; this takes total time $n^{\Oh(k')}$.
Finally, we call \cref{lem:algo-fvs} in order to solve \textsc{Max Weight Induced Forest} with instance $(G,\wei)$. The overall running time is $n^{\Oh(k')}={n^{\Oh(k)}}$.

\subsection{Proof of \texorpdfstring{\cref{lem:ff}}{Lemma 3.1}}

The proof is based on a structural analysis of potential signatures of an optimal solution at a node $t$ of $\cT$.

\paragraph{What lives in a forest.}
Let $F$ be a forest.
A one-vertex component of $F$ is called \emph{trivial}, each other component is \emph{nontrivial}.
The set of vertices of trivial components of $F$ is denoted by $T(F)$.
Let $L(F)$ consist of \emph{leaves} of $F$, i.e., the vertices of degree 1, with the exception that for each component $C$ isomorphic to $K_2$ we include in $L(F)$ exactly one vertex from $C$, {breaking ties arbitrarily.\footnote{For example, by fixing an arbitrary linear order on $V(G)$ and selecting as the vertex in $V(C)\cap L(F)$ the first vertex of $C$ in the ordering.}}
The \emph{skeleton} of $F$, denoted by $\skel{F}$, is $F \setminus (T(F) \cup L(F))$.
Clearly the sets $\skel{F}, T(F), L(F)$ form a partition of $F$, and the set $\limb{F} \coloneqq L(F) \cup T(F)$ is independent.

\paragraph{Skeletons of maximal induced forests.}
Let $G$ be a graph, $\cT = (T, \{X_t\}_{t \in V(T)})$ be a tree decomposition of $G$ with $\mu(\cT) \leq k$, and $t \in V(T)$.
Let $F$ be a maximal induced forest in $G$.

First, {let us show that, even} though $F \cap X_t$ might be large, the size of $S\coloneqq \skel{F} \cap X_t$ is bounded by a linear function of $k$.

\begin{claim}\label{clm:sizeS}
$|S| \leq 8k$.
\end{claim}
\begin{claimproof}
Let $S_1 \subseteq \skel{F} \cap N[X_t]$ consist of {those} components of $\skel{F} \cap N[X_t]$ that intersect $X_t$.
Now define $S_2$ to be the forest obtained from $S_1$ by (i) repeatedly contracting all edges not intersecting $X_t$, and then (ii) removing leaves that are not in $X_t$.
Observe that $S \subseteq S_2$ and every edge of $S_2$ touches $X_t$.
Each vertex that was removed in step (ii) has a unique neighbor in $S_2$; denote the set of these neighbors by $U$.

A component of $S_2$ is \emph{small} if it has at most two vertices.
Let $a$ be the number of small components of $S_2$ and let $S_3$ be obtained from $S_2$ by removing all small components.
A \emph{thread} in $S_3$ is a maximal path with all internal vertices of degree 2.
{Consider a thread in $S_3$} with consecutive vertices $v_0,v_1,\ldots,v_\ell$, where $v_0$ and $v_\ell$ are of degree other than 2 and $\ell \geq 5$.
Let $q = \lfloor (\ell-5)/3 \rfloor$.
We mark edges $v_{3i+2},v_{3i+3}$ for $i \in [0, q]$.
Finally, we remove all vertices $v_j$ for $j \in [2,\ldots,3q+4]$ and add an edge between $v_1$ and $v_{3q+5}$; note that $3q+5 \leq \ell$.
In this step we marked $q+1$ edges and removed $3(q+1)$ vertices.
The marked edges are pairwise non-adjacent (i.e., they form an induced matching).
Furthermore, they are not incident with vertices of degree other than 2 in $S_3$.
Let $S_4$ be obtained from $S_3$ by the application of the step above to each thread of $S_3$.
Define $b \coloneqq (|S_3| - |S_4|)/3$.

Finally, define $c$ to be the number of leaves in $S_4$.
Note that the number of vertices of degree at least 3 in $S_4$ is at most $c$,
thus $S_4$ has at most $2c-1$ threads,
each with at most 3 internal vertices.
Consequently, $S_4$ has at most $2c + (2c-1) \cdot 3 \leq 8c$ vertices.

Summing up, we obtain that
\begin{equation}
|S| \leq |S_2| \leq 2a + 3b + 8c. \label{eq:sizeS}
\end{equation}

Let $A$ be a set of size $a$ containing one vertex per each small component of $S_2$; note that $A \subseteq X_t$.
{For each $v \in A$, as $v \in \skel{F}$, it must be adjacent to some other vertex $u(v)$ of $F$.
Note that for each small component $H$ of $S_2$, the unique component of $S_1$ that contains $H$ intersects $X_t$ in no component other than $H$.}
Thus $M_A = \{ vu(v) ~|~ v \in A\}$ is an induced matching in $G$ of size $a$ that touches $X_t$.

Let $B$ be the set of all edges that were marked while obtaining $S_4$ from $S_3$.
Note that they induce a matching $M_B$ in $G$ of size $b$ that touches $X_t$.

Finally, let $C$ be the set of leaves of $S_4$.
Recall that $C \subseteq X_t$ and the vertices from $C$ are also leaves in $S_2$. Furthermore, as $S_4$ does not contain small components,
the vertices in $C$ form an independent set in $G$.
For each $v \in C$, define a neighbor $u(v)$ of $v$ as follows.
If $v \notin U$, then none of its neighbors got removed in step (ii) and, {since the remaining steps preserve the degrees in $\skel{F}$ of vertices in $(S_4\cap X_t) \setminus U$, we infer that} $v$ is a leaf of $\skel{F}$.
Consequently, $v$ is adjacent to some leaf {$u(v)\in L(F)$}.
If $v \in U$, there exists $u(v) \in \skel{F} \cap N(X_t) \cap N(v)$ that was deleted in step (ii).
We observe that $M_C\coloneqq\{vu(v)~|~v \in C\}$ is an induced matching in $G$ of size $c$ that touches $X_t$.

We claim that $M = M_A \cup M_B \cup M_C$ is an induced matching in $G$ of size $a + b + c$.
It is straightforward to verify that $M$ is a matching of the desired size that touches $X_t$.
The only nontrivial thing to verify is that there are no edges between the vertices inducing $M_B$ and $M_C$. However, recall that none of the endpoints of edges marked in the process of creating $S_4$ from $S_3$ is adjacent to a vertex of degree other than 2, so $M$ is indeed an induced matching.

We know that $a + b + c = |M| \leq k$, which, combined with \eqref{eq:sizeS}, completes the proof.
\end{claimproof}

Let us remark that we did not make too much effort to optimize the multiplicative constant in \cref{clm:sizeS}.
We decided to aim for simplicity as the constant is, anyway, hidden in the $\Oh(\cdot)$ notation in the final running time.

\paragraph{Identifying the impostors.}
Now we need to understand the set $\limb{F} \cap X_t$; recall that it might be arbitrarily large.
By \cref{prop:yolov-misets} we know that the number of ways \emph{maximal} independents sets of $G$ might intersect $X_t$ is polynomially bounded.
We are going to exploit this.
Since $\limb{F}$ is an independent set, it is contained in some maximal independent set $I^*(F)$ in $G$, and $I^*(F) \cap X_t \in \misets_t$.
Consequently, the number of possible choices for $I \coloneqq I^*(F) \cap X_t$ is bounded by $n^{2k+2}$.

However, we still need to be able to filter out the vertices from {$I\setminus \limb{F}$} using some information that can be stored and processed efficiently.
We will show that this is indeed possible.

{Consider a vertex $v \in I \setminus \limb{F}$.}
By the maximality of $F$, there are two possible reasons why $v \notin \limb{F}$:
\begin{enumerate}
    \item $v \in \skel{F}$,
    \item $F \cup \{v\}$ does not induce a forest, i.e., $v$ has two neighbors in the same component of $F$.
\end{enumerate}
Note that the vertices of the first type can be easily detected, as, by \cref{clm:sizeS}, we can afford to enumerate all possibilities for the set $S\coloneqq\skel{F} \cap X_t$ explicitly.
So let us focus on the vertices of the second type; call them \emph{impostors}.
Since $T(F) \cup L(F) = \limb{F} \subseteq I^*(F)$ and $I^*(F)$ is independent, we know that the neighbors of $v$ in $F$ must be in $\widehat{S} \coloneqq N[X_t] \cap \skel{F}$.

It will also be convenient to distinguish vertices from $T(F)$ and vertices from $L(F)$.
Note that the latter ones have \emph{exactly} one neighbor in $\widehat{S}$, while the former ones have no neighbors in $\widehat{S}$.
Summing up, we observe that $\widehat{S}$
\begin{enumerate}[({Q}1)]
    \item 2-dominates all impostors and does not 2-dominate any vertex from $\limb{F} \cap X_t$,
    \item 1-dominates all vertices from $L(F) \cap X_t$,
    \item is non-adjacent to $T(F) \cap X_t$.
\end{enumerate}
Let $Q \subseteq \widehat{S}$ be an  inclusion-wise minimal subset of $\widehat{S}$ that still satisfies (Q1), (Q2), and (Q3).

\begin{claim}\label{clm:sizeQ}
    $|Q| \leq 4k$.
\end{claim}
\begin{claimproof}
Define $L \coloneqq L(F) \cap X_t$ and let $A$ be the set of impostors, i.e., vertices in $I \setminus F$ (see \Cref{fig:l31}).

\begin{figure}[ht]
\centering
\includegraphics[width = \textwidth]{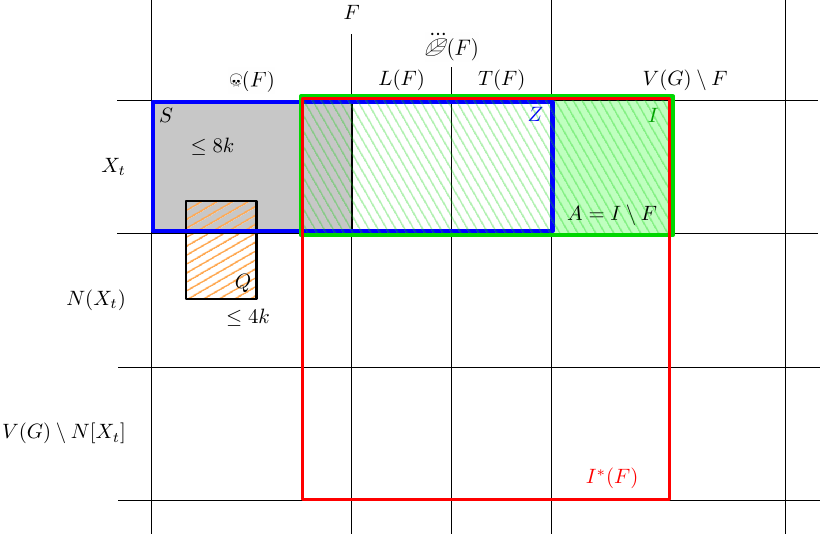}
\caption{Various subsets of $V(G)$ relevant in the proof of \Cref{lem:ff}.}\label{fig:l31}
\end{figure}

Let $Q_1 \subseteq Q$ be a minimal set that 1-dominates $L \cup A$.
From the minimality of $Q_1$ it follows that for every $q \in Q_1$ there exists $m(q) \in L \cup A$ such that $N(m(q)) \cap Q_1 = \{q\}$.
Clearly the set $M_1 \coloneqq\{qm(q)~|~q \in Q_1\}$ is a matching.
Furthermore, the only edges between the vertices of $\bigcup_{e \in M_1} e$ that are not in $M_1$ are between the vertices of $Q_1$.
As $Q_1$ is a forest, it is in particular 2-colorable, hence, there exists an independent set $Q'_1 \subseteq Q_1$ of size at least $|Q_1|/2$.
Hence, $M'_1\coloneqq\{qm(q)~|~q \in Q'_1\}$ is an induced matching of size at least $|Q_1|/2$ that touches $X_t$.

Now let $A'$ be the set of vertices from $A$ that are not 2-dominated by $Q_1$.
We observe that the set $Q_2 \coloneqq Q \setminus Q_1$ 1-dominates $A'$ and, by the minimality of $Q$, is minimal with this property.
Thus, similarly as in the previous paragraph we can conclude that there exists an induced matching of size at least $|Q_2|/2$ that touches $X_t$.

As $Q_1 \cup Q_2 = Q$, one of these sets has at least $|Q|/2$ vertices.
Thus, there exists an induced matching of size at least $|Q|/4$ that touches $X_t$.
As this matching can have at most $k$ edges, we conclude that $|Q| \leq 4k$.
\end{claimproof}

\paragraph{Construction of $\FF_t$.}
Let $S, I, Q$ be the sets defined as above.
Furthermore, let $\widehat{\pi}$ be the partition of $S \cup Q$ corresponding to the connected components of $F \cap V_t$,
i.e., two vertices are in one block of $\widehat{\pi}$ if and only if they are in the same connected component of $F \cap V_t$.

We claim that the quadruple $(S,I,Q,\widehat{\pi})$ uniquely determines the signature $(Z,\pi)$ of $F$ at $t$.

Recall that $Z = F \cap X_t = (\skel{F} \cap X_t) \cup (L(F) \cap X_t) \cup (T(F) \cap X_t)$.
Indeed, we have $\skel{F} \cap X_t = S$ by the definition of $S$.
Next, $L(F) \cap X_t$ consists of these vertices from $I \setminus S$ that are 1-dominated but not 2-dominated by $Q$.
Finally, $T(F)$ are the vertices from $I \setminus S$ that are non-adjacent to $Q$.
Note that the remaining vertices from $I \setminus S$, i.e., {those} that are 2-dominated by $Q$, are impostors.

Now let us argue how to obtain $\pi $ from $\widehat{\pi}$.
We start with $\pi \coloneqq \widehat{\pi}$; we will modify this partition as follows.
First, recall that each $v \in L(F) \cap X_t$ has exactly one neighbor $u$ in $F$, and additionally $u \in Q$.
If $u \in V_t$, then $v$ is in the same component of $F \cap V_t$ as $u$, i.e., we insert $v$ into the block of $\pi$ containing $u$. If $u \notin V_t$, then $v$ is an isolated vertex in $F \cap V_t$ and thus we add a new block $\{v\}$ to $\pi$.
Similarly, for each $v \in T(F) \cap X_t$, we add to $\pi$ a new block $\{v\}$.
Finally, we remove from the blocks of $\pi$ all vertices from $Q \setminus S$.
Note that the last step might result in creating an empty block; we remove such blocks.
This completes the construction of~$\pi$.

From the description above it follows that $(Z,\pi)$ is indeed the signature of $F$ at $t$.
We can exhaustively guess $(S,I,Q,\widehat{\pi})$ which results in at most
\[
n^{8k} \cdot n^{2k+2} \cdot n^{4k} \cdot (12k)^{12k} = (12k)^{12k} \cdot n^{14k+2}
\]
possibilities.
For each choice we construct in polynomial time one pair $(Z,\pi)$ and insert it to $\FF_t$.
Note that some pairs could be immediately discarded, e.g., if $Z$ is not a forest, but this is not necessary so we do not do it.
This completes the proof of \cref{lem:ff}.

\subsection{Proof of \texorpdfstring{\cref{lem:algo-fvs}}{Lemma 3.2}\label{subsec:proof}}

The proof of \cref{lem:algo-fvs} follows the textbook dynamic programming algorithm for \textsc{Max Weight Induced Forest} in bounded-treewidth graphs.
However, the restriction that we can only iterate on elements of $\FF_t$ for $t \in V(T)$ adds some extra technical complicacy to the proof of correctness.

\paragraph{Filling the dynamic programming table.}
We index a dynamic programming table $\Tab$ by the nodes $t \in V(T)$ and the pairs $(Z, \pi) \in \FF_t$ (recall that $Z \subseteq X_t$ and $\pi$ is a partition of $Z$). The intended meaning is as follows:
\begin{description}
\item[($\star$)] The value of $\Tab_t[Z,\pi]$ is the maximum weight of an induced forest in $V_t$ with signature $(Z,\pi)$ at $t$, or $-\infty$ if such a forest does not exist.
\end{description}

Note that since $X_{\root(T)} = \emptyset$, the signature of \emph{every} induced forest in $G$ at $\root(T)$ is $(\emptyset,\emptyset)$.
In particular, $(\emptyset,\emptyset) \in \FF_{\root(T)}$. Thus, if property ($\star$) is satisfied,
the value $\Tab_{\root(T)}[\emptyset,\emptyset]$ is the weight of an optimal solution to \textsc{Max Weight Induced Forest}.

\medskip
Let $t \in V(T)$ be a node and consider an arbitrary $(Z,\pi) \in \FF_t$.
Before we proceed to the detailed description, notice that in some cases we can immediately detect that $(Z,\pi)$ is not a signature of any solution. In particular, this might happen if:
\begin{itemize}
    \item $Z$ contains a cycle, or
    \item there are two adjacent vertices in $Z$ that are in different blocks of $\pi$.
\end{itemize}
In these two cases we immediately set $\Tab_{t}[Z,\pi] = -\infty$.
So from now on assume that $Z$ is an induced forest and $\pi$ \emph{respects the components of $Z$}, i.e., any two vertices in the same component of $Z$ are in the same block of $\pi$.
{How} the value of $\Tab_t[Z,\pi]$ is computed depends on the type of $t$.

\subparagraph{Leaf node.} If $t$ is a leaf, then $X_t = \emptyset$, thus the signature of every induced forest at $t$ is $(\emptyset,\emptyset)$, so in particular $(\emptyset,\emptyset) \in \FF_t$.
We set $\Tab_t[\emptyset,\emptyset]=0$.

\medskip
For all other types of nodes the value of $\Tab_t[Z,\pi]$ depends on the values stored in $\Tab$ for children node(s).
In the description below we use the convention that the maximum over an empty set is $-\infty$.

\subparagraph{Introduce node.} Let $t$ be an introduce node, $t'$ be its child in $T$, and let $v$ be the unique vertex in $X_t \setminus X_{t'}$.

\noindent\emph{Case 1. $v \notin Z$.} We set $\Tab_t[Z,\pi] = \Tab_{t'}[Z,\pi]$.

\noindent\emph{Case 2. $v \in Z$.}
A partition $\pi'$ of $Z \setminus \{v\}$ is \emph{$v$-good} if each vertex from $N(v) \cap Z$ is in a distinct block of $\pi'$.
For a $v$-good partition $\pi'$ of $Z \setminus \{v\}$, by $\pi' \oplus v$ we mean the partition obtained from $\pi'$ by first adding a new block $\{v\}$ and then merging all blocks intersecting $N[v]$ into a single one.
We set
\begin{equation}
\Tab_t[Z,\pi] = \max_{\substack{(Z\setminus \{v\}, \pi') \in \FF_{t'} \\ \text{ s.t. } \pi' \text{ is $v$-good}\\ \text{and } \pi' \oplus v = \pi}} \Tab_{t'}[Z \setminus \{v\}, \pi'] + \wei(v). \label{eq:fillIntroduce}
\end{equation}

\subparagraph{Forget node.}  Let $t$ be a forget node, $t'$ be its child in $T$, and let $v$ be the unique vertex in $X_{t'} \setminus X_{t}$.

Let $\Sigma$ be a set defined as follows.
First, we include $(Z,\pi)$ into $\Sigma$.
Next, for each block $B$ of $\pi$, we include to $\Sigma$ the pair $(Z \cup \{v\}, \pi_B)$, where $\pi_B$ is obtained from $\pi$ by adding $v$ to $B$.
Finally, we add to $\Sigma$ the pair $(Z \cup \{v\}, \pi_{\emptyset})$, where $\pi_{\emptyset}$ is obtained from $\pi$ by adding a new block $\{v\}$.
We set
\begin{equation}
\Tab_t[Z,\pi] = \max_{(Z',\pi') \in \Sigma \cap \FF_{t'}} \Tab_{t'}[Z',\pi']. \label{eq:setForget}
\end{equation}

\subparagraph{Join node.}
Let $t$ be a join node, and let $t',t''$ be its children. Recall that we have $X_t = X_{t'} = X_{t''}$.

Consider two partitions $\pi',\pi''$ of $Z$ that respect the components of $Z$. We define an auxiliary multigraph $\mathsf{Merge}(\pi',\pi'')$ as follows. The vertices of $\mathsf{Merge}(\pi',\pi'')$ are the components of $Z$.
We add an edge between two components if they are in the same block of $\pi'$, and then similarly for $\pi''$; note that we might create parallel edges.
If $\mathsf{Merge}(\pi',\pi'')$ is a simple forest (in particular, does not contain double edges), we define a partition $\pi' \otimes \pi''$ of $Z$ as follows.
For each component $C$ of $\mathsf{Merge}(\pi',\pi'')$, we add a new block to $\pi' \otimes \pi''$ containing the union of the vertex sets of the components of $Z$ in $V(C)$.
If $\mathsf{Merge}(\pi',\pi'')$ is not a forest, then $\pi' \otimes \pi''$ is undefined.

We set
\begin{equation}
\Tab_t[Z,\pi] = \max_{\substack{(Z,\pi') \in \FF_{t'}, \\ (Z,\pi'') \in \FF_{t''},\\\text{s.t. } \pi' \otimes \pi'' = \pi}} \left( \Tab_{t'}[Z,\pi'] + \Tab_{t''}[Z,\pi''] \right) - \wei(Z). \label{eq:fillJoin}
\end{equation}

\paragraph{Running time.}
Recall that the number of nodes is polynomial in $n$, and the number of entries in $\Tab$ for each node $t$ is exactly $|\FF_t|$, where $|\FF_t| = n^{\Oh(k)}$ by \cref{lem:ff}. Let us discuss the time needed to fill a single entry.

The unique entry for a leaf node is filled in constant time.
Each entry in a forget  node is clearly filled in time $n^{\Oh(1)}$.
For an introduce {(resp., join) node}, we need to iterate over elements of $\FF_{t'}$ (resp.~$\FF_{t'} \times \FF_{t''}$), which gives us $n^{\Oh(k)}$ possibilities. For each of them the local computation takes time $n^{\Oh(1)}$.

Thus the total running time of the algorithm is $n^{\Oh(k)}$, as claimed.

\paragraph{Correctness.}
To show that the algorithm is correct, we need to show that the property ($\star$) is indeed satisfied.
We will show it in the following two claims.

\begin{claim}\label{clm:fvs-correctness-Tabatleast}
    Let $F$ be a maximal induced forest in $G$ and let $t \in V(T)$.
    Let $(Z,\pi)$ be the signature of $F$ at $t$.
    Then we have $\Tab_t[Z,\pi] \geq \wei(F \cap V_t)$.
\end{claim}
\begin{claimproof}
    The proof goes by induction on the depth of $t$.
    The base case is when $t$ is a leaf.

    \subparagraph{Leaf node.} As already mentioned, the signature of $F$ at $t$ is $(\emptyset,\emptyset)$, so by \cref{lem:ff} we have $(\emptyset,\emptyset) \in \FF_t$. Thus we correctly set $\Tab_t[\emptyset,\emptyset]=0$.

    \medskip
    Now {assume} that $t$ is not a leaf and the claim holds for all children of $t$.
    \subparagraph{Introduce node.} Let $t'$ be the child of $t$ and let $v$ be the unique vertex in $X_t \setminus X_{t'}$. Recall that $V_t = V_{t'} \cup \{v\}$. We consider two cases.

\smallskip
    \noindent\emph{Case 1. $v \notin F$.}
    This means that $v \notin Z$,
    and thus the signature of $F$ at $t'$ is $(Z,\pi)$. In particular, by \cref{lem:ff}, we know that $(Z,\pi) \in \FF_{t'}$.
    We set $\Tab_t[Z,\pi] = \Tab_{t'}[Z,\pi] \geq \wei(F \cap V_{t'}) = \wei(F \cap V_t)$, where the ``$\geq$'' inequality follows by the induction hypothesis.

\smallskip
    \noindent\emph{Case 2. $v \in F$.}
    Let $(Z',\pi')$ be the signature of $F$ at $t'$.
    By \cref{lem:ff} we have $(Z',\pi') \in \FF_{t'}$.
    We note that $Z' = Z \setminus \{v\}$.
    Let $C$ be the component of $F \cap V_t$ containing $v$,
    and let $C_1,\ldots,C_p$ be the components of $C \setminus \{v\}$ (possibly $p=0$ if $C$ has just one vertex, $v$).
    Note that each $C_i$ contains exactly one neighbor of $v$.
    We note that blocks of $\pi'$ are:
    \begin{itemize}
        \item blocks of $\pi$ corresponding to components of $F \cap V_t$ other than $C$; they do not contain neighbors of $v$, and
        \item $p$ blocks, each containing the vertex set of one $C_i$; each such component contains one neighbor of $v$.
    \end{itemize}
    Thus $\pi'$ is $v$-good and $\pi' \oplus v = \pi$.
    Summing up, $(Z',\pi')=(Z \setminus \{v\},\pi')$ is one of the options considered when we set the value of $\Tab_t[Z,\pi]$ in \eqref{eq:fillIntroduce}.
    Therefore, by the induction hypothesis, we have
    \[
    \Tab_t[Z,\pi] \geq \Tab_{t'}[Z,\pi'] + \wei(v) \geq \wei(F \cap V_{t'}) + \wei(v)= \wei(F \cap V_t),
    \]
    as claimed.

    \subparagraph{Forget node.} Let $t'$ be the child of $t$ and let $v$ be the unique vertex in $X_{t'} \setminus X_{t}$. Recall that $V_t = V_{t'}$. We consider two cases.

\smallskip
    \noindent\emph{Case 1. $v \notin F$.}
    Note that in this case the signature of $F$ at $t'$ is $(Z,\pi)$ and thus by \cref{lem:ff} this pair is in $\FF_{t'}$.
    Note that we have $(Z,\pi) \in \Sigma \cap \FF_{t'}$, i.e., in the algorithm we set
    \[
    \Tab_t[Z,\pi] \geq \Tab_{t'}[Z, \pi] \geq \wei(F \cap V_{t'}) = \wei(F \cap V_t),
    \]
    where the second ``$\geq$'' follows by the induction hypothesis.

\smallskip
    \noindent\emph{Case 2. $v \in F$.}
    Let $(Z',\pi')$ be the signature of $F$ at $t'$; by \cref{lem:ff} we have $(Z',\pi') \in \FF_{t'}$.
    Note that we have $Z' = Z \cup \{v\}$.
    Let $B'$ be the block of $\pi'$ containing $v$ and denote $B' \setminus \{v\}$ by $B$.
    Since $v\in V_t\setminus X_t$, the partition $\pi$ is obtained from $\pi'$ by replacing $B'$ by $B$, or by removing $B'$ if $B' = \{v\}$.
    In other words, we have $\pi' = \pi_B$ (with $B = \emptyset$ if $B' = \{v\}$).
    Therefore, in the algorithm we set
    \[
    \Tab_t[Z,\pi] \geq \Tab_{t'}[Z \cup \{v\}, \pi_B] \geq \wei(F \cap V_{t'}) = \wei(F \cap V_t),
    \]
    where, again, the second ``$\geq$'' follows by the induction hypothesis.

    \subparagraph{Join node.}
    Let $t',t''$ be the children of $t$; recall that we have $X_t = X_{t'} = X_{t''}$. Furthermore, it holds that $V_t = V_{t'} \cup V_{t''}$ and $V_{t'} \cap V_{t''} = X_t$.

    Let $(Z,\pi')$ (resp.~$(Z,\pi'')$) be the signature of $F$ at $t'$ (resp.~at $t''$). By \cref{lem:ff} we have $(Z,\pi') \in \FF_{t'}$ and $(Z,\pi'') \in \FF_{t''}$.
     Note that both $\pi'$ and $\pi''$ respect the components of $Z$.

    We claim that $\pi' \otimes \pi'' = \pi$.
    First, let us show that $\mathsf{Merge}(\pi',\pi'')$ is a forest.
    For contradiction, suppose it contains a cycle $C$.
    Recall that $C$ has two types of edges: defined by blocks of $\pi'$ and defined by blocks of $\pi''$.
    Each edge corresponds to a path in $F$ with endpoints in appropriate components of $Z$. By connectivity of each component of $Z$, we observe that $C$ defines a cycle in $F$, which is a contradiction as $F$ is a forest.

    Note that two vertices of $Z$ are in the same connected component of $F\cap V_t$ if and only if one of the following happens: they are in the same connected component of $F \cap V_{t'}$, or they are in the same connected component of $F \cap V_{t''}$ (both possibilities apply if they are in the same component of $Z$).
    Thus indeed we have that $\pi' \otimes \pi'' = \pi$.

    This means that in the algorithm we set
    \begin{align*}
    \Tab_t[Z,\pi] \geq & \ \Tab_{t'}[Z, \pi'] + \Tab_{t''}[Z,\pi''] - \wei(Z) \geq  \ \wei(F \cap V_{t'}) + \wei(F \cap V_{t''}) - \wei(Z) \\ = & \  \wei(F \cap V_{t}) + \wei(F \cap X_t) - \wei(Z) = \wei(F \cap V_t),
    \end{align*}
    where the second ``$\geq$'' follows by the induction hypothesis.
    This completes the proof.
\end{claimproof}

\begin{claim}\label{clm:fvs-correctness-Tabatmost}
    Let $t \in V(T)$ and let $(Z,\pi) \in \FF_t$ such that $\Tab_t[Z,\pi] \neq -\infty$.
    Then there exists an induced forest $F$ in $V_t$ with signature $(Z,\pi)$ at $t$ satisfying $\wei(F) \geq \Tab_t[Z,\pi]$.
\end{claim}
\begin{claimproof}
    The proof again goes by induction on the depth of $t$.
    The base case is that $t$ is a leaf node.

\subparagraph{Leaf node.}
    If $t$ is a leaf node, then the only choice for $(Z,\pi)$ is $(\emptyset,\emptyset)$, and we set \hbox{$\Tab_t[\emptyset,\emptyset]=0$}.
    Note that $F = \emptyset$ satisfies the desired properties.

    \medskip
    So from now on we assume that $t$ is not a leaf and the claim holds for all children of $t$. Again, the argument depends on the type of $t$.

\subparagraph{Introduce node.} Let $t$ be an introduce node, $t'$ be its child in $T$, and let $v$ be the unique vertex in $X_t \setminus X_{t'}$.
Since $\Tab_t[Z,\pi] \neq -\infty$, its value was set either in Case 1 or in Case 2 of the algorithm.
Consider these cases separately.

\smallskip
\noindent\emph{Case 1.} This means that $\Tab_t[Z,\pi] = \Tab_{t'}[Z,\pi]$.
By the induction hypothesis there is an induced forest $F \subseteq V_{t'}$ with signature $(Z,\pi)$ at $t'$ with $\wei(F) \geq \Tab_{t'}[Z,\pi]$.
Notice that the signature of $F$ at $t$ is $(Z,\pi)$.
Moreover, $F$ satisfies $\wei(F)  \geq \Tab_{t'}[Z,\pi] = \Tab_{t}[Z,\pi]$, as claimed.

\smallskip
\noindent\emph{Case 2.} This means that $\Tab_t[Z,\pi] = \Tab_{t'}[Z \setminus \{v\},\pi'] + \wei(v)$ for some $(Z \setminus \{v\}, \pi') \in \FF_{t'}$, such that $\pi'$ is $v$-good and $\pi' \oplus v = \pi$, as set in \eqref{eq:fillIntroduce}.
As $\Tab_t[Z,\pi] \neq -\infty$, we also have $\Tab_{t'}[Z \setminus \{v\},\pi'] \neq -\infty$, so by the induction hypothesis there exists an induced forest $F' \subseteq V_{t'}$ with signature $(Z \setminus \{v\}, \pi')$ at $t'$, satisfying $\wei(F') \geq \Tab_{t'}[Z \setminus \{v\}, \pi']$.

We aim to show that $F \coloneqq F' \cup \{v\}$ satisfies the statement of the claim.
First we need to argue that $F$ is an induced forest.
For contradiction suppose that $F$ contains a cycle $C$.
Since $F'$ is acyclic, we observe that $v$ belongs to $C$ and $C-v$ is contained in one component of $F'$.
As the neighborhood of $v$ in $V_{t'}$ is contained in $X_{t'}$, we observe that the neighbors of $v$ on $C$, say $u$ and $w$, are in $Z \setminus \{v\}$.
Consequently, $u,w$ are in one block of $\pi'$ and so $\pi'$ is not $v$-good, a contradiction.

Now let us argue that the signature of $F$ at $t$ is $(Z,\pi)$.
Clearly $F \cap V_t = F \cap \left(V_{t'} \cup \{v\}\right) = (Z \setminus \{v\}) \cup \{v\} = Z$.
On the other hand, it is straightforward to verify that the partition of $Z$ corresponding to components of $F$ is exactly $\pi' \oplus v = \pi$.
Finally, note that $\wei(F) = \wei(F') + \wei(v) \geq  \Tab_{t'}[Z \setminus \{v\}, \pi'] + \wei(v) =  \Tab_{t}[Z, \pi]$.

\subparagraph{Forget node.}  Let $t'$ be the child of $t$ and let $v$ be the unique vertex in $X_{t'} \setminus X_{t}$.
As $\Tab_t[Z,\pi] \neq -\infty$, its value was set to $\Tab_{t'}[Z',\pi']$ for some $(Z',\pi') \in \Sigma \cap \FF_{t'}$, as in \eqref{eq:setForget}.

Let $F$ be an induced forest in $V_{t'}$ with signature $(Z',\pi')$ at $t'$, such that $\wei(F) \geq \Tab_{t'}[Z',\pi']$.
Consider two cases.

\smallskip
\noindent\emph{Case 1. $v \notin F$.} Then by the definition of $\Sigma$ we observe that $Z' = Z$ and $\pi' = \pi$.
The signature of $F$ at $t$ is $(Z,\pi)$.

\smallskip
\noindent\emph{Case 2. $v \in F$.} Then by the definition of $\Sigma$ we have that $Z' = Z \cup \{v\}$ and $\pi' = \pi_B$ for some block $B$ of $\pi$ or $\pi' = \pi_\emptyset$.
Note that the signature of $F$ at $t$ is $(Z,\pi)$.

\smallskip
In both cases we have $\wei(F) \geq \Tab_{t'}[Z',\pi'] = \Tab_t[Z,\pi]$, so $F$ satisfies the statement of the claim.

\subparagraph{Join node.} Let $t',t''$ be the children of $t$; recall that we have $X_t = X_{t'} = X_{t''}$.
Since $\Tab_t[Z,\pi] \neq -\infty$, we have $\Tab_t[Z,\pi] = (\Tab_{t'}[Z,\pi'] + \Tab_{t''}[Z,\pi'']) - \wei(Z)$
for some $(Z,\pi') \in \FF_{t'}$ and $(Z,\pi'') \in \FF_{t''}$ such that $\pi' \otimes \pi'' = \pi$,
as set in \eqref{eq:fillJoin}.

Let $F'$ (resp.~$F''$) be an induced forest in $V_{t'}$ (resp.~$V_{t''}$) with signature $(Z,\pi')$ at $t'$ (resp.~$(Z,\pi'')$ at $t''$)
such that $\wei(F') \geq \Tab_{t'}[Z,\pi']$ (resp.~$\wei(F'') \geq \Tab_{t''}[Z,\pi'']$).

We claim that $F\coloneqq F' \cup F''$ is an induced forest in $V_t$.
For contradiction, suppose that $F$ {contains} a cycle.
Note that such a cycle cannot be contained in $Z$ as we know that $Z$ induces a forest (since the entries of $\Tab$ corresponding to $Z$ have finite values).
So in particular $F$ contains a cycle $C$ whose intersection with every component of $Z$ is a subpath of $C$.
Let $C'$ be the multigraph obtained from $C$ by (1) contracting every edge with both endvertices in $Z$, and then (2) substituting every maximal path with all interior vertices outside $Z$ by a single edge.
Note that this step might {have created} multiple edges, but no loops, since both $F'$ and $F''$ are induced forests and $Z$ does not contain cycles.
Thus $C'$ is a cycle whose every vertex corresponds to a distinct component of $Z$.
Every edge of $C'$ indicates that either $F'$ or $F''$ contains a path with endpoints in particular components of $Z$.
Therefore $C'$ is a sub(multi)graph of $\textsf{Merge}(\pi',\pi'')$, so in particular $\textsf{Merge}(\pi',\pi'')$ is not a simple forest.
Consequently, $\pi' \otimes \pi''$ is not defined, so it cannot be equal to $\pi$.

Clearly $F \cap X_t = Z$. Furthermore, the partition of $Z$ corresponding to components of $F$ is precisely $\pi = \pi' \otimes \pi''$,
so $(Z,\pi)$ is the signature of $F$ at $t$.

Finally we notice that $\wei(F) = \wei(F') + \wei(F'') - \wei(F' \cap F'') = \wei(F') + \wei(F'') - \wei(Z) \geq \Tab_{t'}[Z,\pi'] + \Tab_{t''}[Z,\pi''] - \wei(Z) = \Tab_t[Z,\pi]$.
This completes the proof of claim.
\end{claimproof}

\bigskip
Now let us argue that \cref{clm:fvs-correctness-Tabatleast,clm:fvs-correctness-Tabatmost} are sufficient to prove the correctness of our algorithm.
Let $F$ be an optimal solution. We can safely assume that $F$ is a maximal induced forest.
Recall that its signature at $\root(T)$ is $(\emptyset,\emptyset)$.
By \cref{clm:fvs-correctness-Tabatleast} we have that $\Tab_{\root(T)}[\emptyset,\emptyset] \geq \wei(F)$.
Now suppose {for a contradiction} that $\Tab_{\root(T)}[\emptyset,\emptyset] > \wei(F)$.
Thus, by \cref{clm:fvs-correctness-Tabatmost}, there is an induced forest $F^*$ in $G$ with $\wei(F^*) \geq \Tab_{\root(T)}[\emptyset,\emptyset] > \wei(F)$. However, this contradicts the optimality of $F$.
Therefore, {$\Tab_{\root(T)}[\emptyset,\emptyset]$ is indeed} the weight of an optimal solution.
This completes the proof of \cref{lem:algo-fvs} and thus of \cref{thm:fvs}.

\section{Packing independent induced subgraphs}\label{sec:blob}
Recall that given a graph $G$ and a family $\cH =\{H_j\}_{j\in J}$ of  {connected} subgraphs of $G$, we denote by $G^\circ[\cH]$ the graph with {vertex set $J$,} in which two distinct {vertices} $i,j\in J$ are adjacent if and only if $H_i$ and $H_j$ either have a vertex in common or there is an edge in $G$ connecting them.
The starting point of our investigations will be the following result of Dallard et al.~\cite{dallard2022firstpaper}.

\begin{sloppypar}
\begin{lemma}[Dallard et al.~\cite{dallard2022firstpaper}]\label{tin-of-G(H)}
Let $G$ be a graph, let \hbox{$\mathcal{T} = (T, \{X_t\}_{t\in V(T)})$} be a tree decomposition of $G$, and let $\HH =\{H_j\}_{j\in J}$ be a finite family of connected non-null subgraphs of $G$.
Then \hbox{$\mathcal{T}' = \big(T, \{\Bag'_t\}_{t\in V(T)}\big)$} with $\Bag'_t = \{{j\in J}\colon V(H_j) \cap X_t \neq \emptyset\}$ for all $t \in V(T)$ is a tree decomposition of $G^\circ[\cH]$ such that $\alpha(\mathcal{T}')\le \alpha(\mathcal{T})$.
\end{lemma}
\end{sloppypar}

\begin{sloppypar}
\begin{corollary}[Dallard et al.~\cite{dallard2022firstpaper}]\label{tin-of-G(H)-corollary}
Let $G$ be a graph and let $\HH =\{H_j\}_{j\in J}$ be a finite family of connected non-null subgraphs of $G$.
Then $\tin(G^\circ[\cH])\le \tin(G)$.
\end{corollary}
\end{sloppypar}

\subsection{Tree-independence number and induced matching treewidth of \texorpdfstring{$G^\circ[\cH]$}{G\textdegree[\unichar{"0210B}]}}

In this section we prove results similar to \cref{tin-of-G(H)-corollary}, but with a focus on induced matching treewidth.
Note that while in the statement of \cref{tin-of-G(H)-corollary}, the family $\HH =\{H_j\}_{j\in J}$ of subgraphs of $G$ may be a multiset, in the first part of the next lemma we restrict ourselves to a \emph{set} of subgraphs, with no repetitions allowed.

\begin{lemma}
\label{lem:H-G}
Let $G$ be a graph and {$\HH =\{H_j\}_{j\in J}$} be a set of connected non-null subgraphs of $G$.
Then, $\yw(G^\circ[\cH])\le \yw(G)$.
Furthermore, if all graphs in $\mathcal{H}$ have at least two vertices, then $\tin(G^\circ[\cH])\le \yw(G)$, even if $\mathcal{H}$ is a multiset.
\end{lemma}

\begin{proof}
Let $k = \yw(G)$ and consider a tree decomposition $\cT=(T,\{X_t\}_{t\in V(T)})$ of the graph $G$ such that $\ywT(\cT) = k$.
We want to construct a tree decomposition $\cT'=(T,\{X'_t\}_{t\in V(T)})$ for $G^\circ[\cH]$ such that $\mu_{G^\circ[\cH]}(\cT')\leq k$.
{Note that $V(G^\circ[\cH]) = J$.}
For each $t\in V(T)$, let $X'_t = \{{j\in J}\colon V(H_j)\cap X_t\neq \emptyset\}$.
By \Cref{tin-of-G(H)}, $\cT'$ is a tree decomposition of the graph $G^\circ[\cH]$.
We now show that $\mu_{G^\circ[\cH]}(\cT')\leq k$.
This will establish $\yw(G^\circ[\cH])\le \yw(G)$.
Consider a bag of $\cT'$, say $X'_t$ for some $t\in V(T)$, and suppose for a contradiction that there is an induced matching $M$ in the graph $G^\circ[\cH]$ with cardinality $k+1$ such that each edge of $M$ has an endpoint in $X'_t$.
{Then, there exists a set $I\subseteq X_t'$ with $|I| = k+1$ such that $I$ contains precisely one endpoint of each edge of $M$.
Consider an arbitrary edge $\{i,i'\}$ of $M$ such that $i\in I$.}
Since {$i\in X'_t$}, there exists a vertex $v_i\in X_t\cap V(H_i)$.
Since {$\{i,i'\}$} is an edge in $G^\circ[\cH]$, the union of the graphs $H_i$ and {$H_{i'}$} is a connected subgraph of $G$ containing at least two vertices; in particular, no vertex in {$H_i\cup H_{i'}$} is isolated.
Let $v_{i'}$ be a neighbor of $v_i$ in $H_i\cup H_{i'}$.
Since $M$ is an induced matching in $G^\circ[\cH]$, note that if {$\{i,i'\}, \{j,j'\}\in M$ with $i,j\in I$ and $i\neq j$,} then there is no edge in $G$ connecting a vertex of $H_i\cup H_{i'}$ to a vertex of $H_j\cup H_{j'}$.
Since for each {$i\in I$}, $v_iv_{i'}$ is an edge of $H_i\cup H_{i'}$, we conclude that $\{v_iv_{i'}\colon {i\in I}\}$ is an induced matching in $G$.
Furthermore, each of the edges in this matching has an endpoint in $X_t$.
This implies that $\ywT(\cT) > k$, a contradiction.

Assume now that each graph in $\mathcal{H}$ has at least two vertices, but repetitions of subgraphs in $\mathcal{H}$ are allowed.
We now show that $\tin(G^\circ[\cH])\le \yw(G)$.
Consider the same tree decomposition {of} $G^\circ[\cH]$ as before.
In this case, we can show, using a similar argument as above, that $\alpha_{G^\circ[\cH]}(\cT') \le k$, which will establish $\tin(G^\circ[\cH])\le \yw(G)$.
Consider a bag of $\cT'$, say $X'_t$ for some $t\in V(T)$, and suppose for a contradiction that there is an independent set $I$ in the graph $G^\circ[\cH]$ with cardinality $k+1$ such that $I\subseteq X'_t$.
Consider an arbitrary {$i\in I$.}
Since {$i\in X'_t$}, there exists a vertex $v_i\in X_t\cap V(H_i)$.
Since the graph $H_i$ is connected and has at least two vertices, there exists a vertex {$v_{i'}$} that is adjacent to $v_i$ in $H_i$.
Since for each {$i\in I$}, the edge {$v_iv_{i'}$} is an edge of $H_i$, the fact that $I$ is an independent set in $G^\circ[\cH]$ and the definition of the graph  $G^\circ[\cH]$ imply that $\{v_i{v_{i'}\colon i\in I}\}$ is an induced matching in $G$.
Furthermore, each of the edges in this matching has an endpoint in $X_t$.
This implies that $\mu_G(\cT) > k$, a contradiction.
\end{proof}

\begin{remark}
In the first part of \Cref{lem:H-G}, the assumption that $\cH$ is a set is necessary.
To see this, fix a positive integer $n$, let $G$ be the complete bipartite graph $K_{n,n}$, and let $\cH$ be the {multiset} of connected subgraphs of $G$ consisting of two copies of each one-vertex induced subgraph of $G$.
Then $G^\circ[\cH]$ is the graph obtained from two disjoint copies of the graph $nK_2$ by adding all possible edges between them, where $nK_2$ denotes the disjoint union of $n$ copies of $K_2$.
Clearly, $\yw(G) = 1$, but, as we will show in \Cref{sec:structural} (see \cref{prop:matching-biclique}), $\yw(G^\circ[\cH])\geq n$.

The same construction also shows that the assumption that all graphs in $\mathcal{H}$ have at least two vertices is necessary in the second part of \Cref{lem:H-G}.
\end{remark}

As the proof of \Cref{lem:H-G} is algorithmic, we obtain the following corollary.

\begin{corollary}\label{coro:compute-td-gH}
Given a graph $G$, a family $\cH$ of connected non-null subgraphs of $G$, and a tree decomposition $\cT=(T,\{X_t\}_{t \in V(T)})$ of $G$, in time polynomial in $|V(G)|, |\cH|,$ and $|V(T)|$ we can compute a tree decomposition $\cT'$ of $G^\circ[\cH]$ such that $\mu_{G^\circ[\cH]}(\cT')\leq\mu_G(\cT)$ if $\mathcal{H}$ is a set, and $\alpha_{G^\circ[\cH]}(\cT')\leq\mu_G(\cT)$ if all graphs in $\mathcal{H}$ have at least two vertices.
\end{corollary}

\subsubsection{Algorithmic consequences}

Let us now explain how \cref{lem:H-G} can be used algorithmically.
Given a graph $G$ and a family $\cH =\{H_j\}_{j\in J}$ of subgraphs of $G$, a subfamily $\cH'$ of $\cH$ is said to be an \emph{independent $\cH$-packing} in $G$ if every two graphs in $\cH'$ are vertex-disjoint and there is no edge { in $G$} between them. In other words, $\cH'$ corresponds to an independent set in the graph $G^\circ[\cH]$.
Assume now that the subgraphs in $\cH$ are equipped with a weight function $\wei\colon J\to \mathbb{Q}_+$ assigning weight $\wei_j$ to each subgraph $H_j$.
For any set $J' \subseteq J$, we define the \emph{weight} of the family $\cH' = \{H_j\}_{j\in J'}$ as the sum $\sum_{j\in J'}\wei_j$.
In particular, the weight of any independent $\cH$-packing is well-defined.

Given a graph $G$, a finite family $\cH = \{H_j\}_{j\in J}$ of connected non-null subgraphs of $G$, and a weight function $\wei\colon J\to \mathbb{Q}_+$ on the subgraphs in $\cH$, the \textsc{Max Weight Independent Packing} problem asks to find an independent $\cH$-packing in $G$ of maximum weight.
A natural special case of the problem is when for each $j \in J$, the value of $\wei_j$ is equal to the number of vertices of $H_j$; we call this restriction \textsc{Max Independent Packing}~\cite{MR2190818}.

It is straightforward to verify that  \textsc{Max Weight Independent Packing} for $G$ and $\cH$ is equivalent to \textsc{MWIS} in $G^\circ[\cH]$, where the weights of vertices of $G^\circ[\cH]$ correspond to the weights of subgraphs in $\cH$.

\begin{sloppypar}
\begin{theorem}\label{max-weight-independent-subgraph-packing-for-yw}
Given a graph $G$ and a finite family $\mathcal{H} = \{H_j\}_{j \in J}$ of connected non-null subgraphs of $G$, and a weight function $\wei\colon J \to \Q_+$, the \textsc{Max Weight Independent Packing} problem can be solved in time polynomial in $|V(G)|$ and $|\cH|$ for graphs of bounded induced matching treewidth.
\end{theorem}
\end{sloppypar}
\begin{proof}[Sketch of proof.]
The proof follows the reasoning of Dallard et al.~\cite{dallard2022firstpaper} for tree-independence number, so let us just sketch it.

Let $G$ and $\cH$ be as in the statement and let $\yw(G) \leq k$, where $k$ is a constant.
First we preprocess the input family $\mathcal{H} = \{H_j\}_{j \in J}$ by keeping at most one copy of each subgraph of $G$.
Indeed, if $H_i = H_j$ {for two} distinct $i,j\in J$ such that $\wei(i)\le \wei(j)$, then not both $H_i$ and $H_j$ can appear in an independent $\cH$-packing, and discarding $H_i$ from the family $\mathcal{H}$ {results} in an equivalent instance.
After this preprocessing step, the resulting family $\mathcal{H}' = \{H_j\}_{j \in J'}$ is a set..
Thus, by \cref{lem:H-G}, we have $\yw(G^\circ[\cH'])\le \yw(G)\le k$.
Finally, we compute $G^\circ[\cH']$ and solve the \textsc{MWIS} problem on $G^\circ[\cH']$ (and the restriction of $\wei$ to $J'$) using \Cref{thm:yolovmwis}.
Since $k$ is a constant, the overall running time is polynomial in $|V(G)|$ and $|\cH|$.
\end{proof}

A class $\cC$ of graphs is \emph{weakly hyperfinite} if for every $\epsilon >0$ there is $c(\epsilon) \in \N$
such that in every graph $G \in \cC$ there is a subset $X \subseteq V(G)$ of at least $(1-\epsilon)|V(G)|$ vertices such that every connected component of $G-X$ has at most $c(\epsilon)$ vertices~\cite[Section 16.2]{DBLP:books/daglib/0030491}.
It turns out that every class that is closed under vertex and edge deletions and admits sublinear balanced separators is weakly hyperfinite. Many well-known classes of sparse graphs are weakly hyperfinite, e.g., graphs of bounded treewidth, planar graphs and, more generally, graphs of bounded genus. In fact, all proper minor-closed classes are weakly hyperfinite.

\begin{theorem}\label{cor:pack-weakly-hyperfinite}
Let $\cC$ be a weakly hyperfinite class of graphs closed under vertex deletion and disjoint union operations.
Let $k \in \N$ and $\epsilon > 0$ be fixed.
Given a graph $G$ with induced matching treewidth at most $k$,
in polynomial time we can find a set $F \subseteq V(G)$ such that:
\begin{enumerate}
\item $G[F] \in \cC$,
\item the size of $F$ is at least $(1-\epsilon) \textsf{OPT}$, where  $\textsf{OPT}$ is the size of a largest set satisfying the first condition.
\end{enumerate}
\end{theorem}
\begin{proof}[Sketch of proof.]
We again only provide a sketch of proof and refer the reader to Gartland et al.~for further details~\cite{gartland2020finding}.
Let $F^*$ be an optimal solution. Since $\cC$ is weakly hyperfinite, there exists a set $F' \subseteq F^*$ such that $|F'| \geq (1-\epsilon) |F^*|$ and every component of $F'$ has at most $c(\epsilon)$ vertices.
As $\cC$ is closed under vertex deletion, we observe that $F' \in \cC$.
Let $\cH$ be the family of induced subgraphs of $G$ of size at most $c(\epsilon)$ that belong to $\cC$.
Note that $\cH$ can be enumerated in time polynomial in $V(G)$ as $c(\epsilon)$ is a constant.

We call the algorithm from \cref{max-weight-independent-subgraph-packing-for-yw} for $G, \cH$, and the weight function defined as the number of vertices for each of $\cH$. Let $F$ be an optimal solution for this instance.
As $\cC$ is closed under disjoint union, we observe that $F \in \cC$. Furthermore, we have $|F| \geq |F'| \geq (1-\epsilon)|F^*|$.
\end{proof}

As bounded-treewidth graphs form a weakly hyperfinite class, we immediately obtain the following corollary.

\ptas*

\subsection{Tree-independence number and induced matching treewidth of graph powers}

Now let us turn our attention to powers of graphs of bounded induced matching treewidth.

\begin{sloppypar}
\begin{lemma}\label{equality-between-G^k+2d-and-G^k(H)}
Let $G$ be a graph, and $k$ and $d$ be positive integers.
For $v\in V(G)$, let $H_v$ be the subgraph of $G$ induced by the vertices at distance at most $d$ from $v$, and let $\mathcal{H} = \{H_v\}_{v\in V(G)}$.
Then, the graphs $G^{k+2d}$ and $(G^k)^\circ[\mathcal{H}]$ are isomorphic.
\end{lemma}
\end{sloppypar}

\begin{proof}
Let us first observe that for every vertex $v\in V(G)$, the graph $H_v$ is a connected subgraph of $G^k$.
Hence, the graph $(G^k)^\circ[\mathcal{H}]$ is well-defined.

Note that the graphs $G^{k+2d}$ and $(G^k)^\circ[\mathcal{H}]$ both have vertex sets of the same size.
Moreover, there is a bijection that maps each $v\in V(G^{k+2d})$ to the vertex corresponding to $H_v$ in $(G^k)^\circ[\mathcal{H}]$. For simplicity of notation, we denote both {vertex} sets by $V(G)$.

We now show that for any $u,v\in V(G)$, $uv\in E(G^{k+2d})$ if and only if $uv\in E((G^k)^\circ[\mathcal{H}])$, which will allow us to conclude that the graphs $G^{k+2d}$ and $(G^k)^\circ[\mathcal{H}]$ are isomorphic.

{Suppose} first that $u$ and $v$ are adjacent in $(G^k)^\circ[\mathcal{H}]$.
Then the graphs $H_u$ and $H_v$ either have a vertex in common, or they are vertex-disjoint but there is an edge between them in $G^k$.
If $H_u$ and $H_v$ have a vertex in common, say $x\in V(H_u)\cap V(H_v)$, then $\dist_G(u,v)\le \dist_G(u,x)+ \dist_G(x,v) \le 2d\le k+2d$, and hence $u$ and $v$ are adjacent in $G^{k+2d}$.
If $H_u$ and $H_v$ are vertex disjoint but there exist vertices $x\in V(H_u)$ and $y\in V(H_v)$ such that $xy\in E(G^k)$, then $\dist_G(u,v)\le \dist_G(u,x)+ \dist_G(x,y) + \dist_G(y,v)\le d+k+d= k+2d$, and hence $u$ and $v$ are adjacent in $G^{k+2d}$.

{Suppose} now that $u$ and $v$ are adjacent in $G^{k+2d}$.
Then $\dist_G(u,v)\le k+2d$.
We want to show that the graphs $H_u$ and $H_v$ either have a vertex in common, or they are vertex-disjoint but there is an edge between them in $G^k$.
If $u$ and $v$ {are} at distance at most $2d$ in $G$, then there exists a vertex $x\in {V(G)}$ that is at distance in $G$ at most $d$ from each of $u$ and $v$, and hence the graphs $H_u$ and $H_v$ have a vertex in common.
We may thus assume that $\dist_G(u,v)\ge 2d+1$ or, equivalently, that the graphs $H_u$ and $H_v$ are vertex-disjoint.
Since $2d+1\le \dist_G(u,v)\le k+2d$, there exists a $u$-$v$-path $P = (u = u_0,u_1,\ldots, u_q = v)$ in $G$ such that $2d+1\le q\le k+2d$.
This implies that $1\le \dist_G(u_d,u_{q-d})\le k$ and hence the vertices $u_d$ and $u_{q-d}$ are adjacent in $G^k$.
Since $u_d$ is at distance at most $d$ from $u$, it belongs to $H_u$.
Similarly, $u_{q-d}$ belongs to $H_v$.
Thus, $u_du_{q-d}\in E(G^k)$, and hence is an edge between the subgraphs $H_u$ and $H_v$.
This shows that $u$ and $v$ are adjacent in $(G^k)^\circ[\mathcal{H}]$.
\end{proof}

\begin{sloppypar}
\begin{lemma}\label{lem:Gr+2}
Let $G$ be a graph with at least one edge and $r$ a positive integer.
Then
\[\yw(G^{r+2})\le \tin(G^{r+2})\le \yw(G^r)\le \tin(G^r)\,.\]
\end{lemma}
\end{sloppypar}

\begin{proof}
Since $\yw(H)\le \tin(H)$ for every graph $H$, it suffices to show that $\tin(G^{r+2})\le \yw(G^r)$.
{ Suppose} first that $G$ is connected.
By \Cref{equality-between-G^k+2d-and-G^k(H)} (applied with $k = r$ and $d = 1$), the graph $G^{r+2}$ is isomorphic to the graph $(G^r)^\circ[\mathcal{H}]$, where $\mathcal{H} = \{H_v\}_{v\in {V(G)}}$ and $H_v$ is the subgraph of $G$ induced by the vertices at distance at most {$1$} from $v$.
Therefore, $\yw(G^{r+2})=\yw((G^r)^\circ[\mathcal{H}])$.
By definition, each subgraph in $\mathcal{H}$ is connected.
Furthermore, since $G$ is a connected graph with at least two vertices, every vertex $v\in V(G)$ has a neighbor and hence $H_v$ has at least two vertices.
Hence, by \Cref{lem:H-G}, $\tin((G^r)^\circ[\mathcal{H}])\leq \yw(G^r)$, and we conclude $\tin(G^{r+2})\leq \yw(G^r)$.

Now, let us consider the general case.
Let $\ell = \yw(G^r)$.
Since $G$ has an edge, so does $G^r$ and consequently $\ell\ge 1$.
By the first case, each nontrivial connected component $C$ of $G$ satisfies $\tin(C^{r+2})\le \yw(C^r)\le \ell$.
Since $\ell\ge 1$, each component $C$ of $G$ with only one vertex  satisfies $\tin(C^{r+2})\le \ell$.
By combining tree decompositons with independence number at most $\ell$ of each component of $G^{r+2}$, we obtain a tree decomposition of $G^{r+2}$ with independence number at most $\ell$.
Therefore, $\tin(G^{r+2})\le \ell$.
\end{proof}

\Cref{lem:Gr+2} is a significant generalization of a result of Duchet~\cite{MR778751}, who proved that for every positive integer $r$, if $G^r$ is chordal, then so is $G^{r+2}$.
This corresponds to the case when $\tin(G^r) \le 1$.

\begin{corollary} \label{coro:ywGr+2}
For any graph $G$ and positive integer $r$, we have $\tin(G^{r+2})\le \tin(G^r)$
and
$\yw(G^{r+2})\le \yw(G^r)$.
\end{corollary}

\begin{proof}
If $G$ has no vertices, then $\tin(G^{r}) = \tin(G^{r+2}) = \yw(G^{r}) = \yw(G^{r+2}) = 0$ and the statement holds.
Assume now that $G$ is non-null.
If $G$ is edgeless, then so are $G^r$ and $G^{r+2}$; hence $\tin(G^{r}) = \tin(G^{r+2}) = 1$ and
$\yw(G^{r}) = \yw(G^{r+2}) = 0$ and the statement holds.
Finally, assume that $G$ has at least one edge.
Then, by \Cref{lem:Gr+2} we have $\yw(G^{r+2})\le \tin(G^{r+2})\le \yw(G^r)\le \tin(G^r)$, which implies both stated inequalities.
\end{proof}

\begin{corollary} \label{coro:ywGk}
Let $G$ be a graph and $r$ be a positive odd integer. Then $\tin(G^{r})\le \tin(G)$ and $\yw(G^r)\leq \yw(G)$.
\end{corollary}

\Cref{coro:ywGk} generalizes a result due to Balakrishnan and Paulraja~\cite{MR704427} stating that the class of chordal graphs is closed under taking odd powers.

\Cref{lem:Gr+2} and a straightforward induction on $r$ also implies the following.

\begin{corollary} \label{coro:ywGr}
Let $G$ be a graph and $r\ge 3$ an odd integer. Then $\tin(G^{r})\le \yw(G)$.
\end{corollary}

Furthermore, as all our proofs are algorithmic, we obtain the following.

\begin{corollary}\label{tin-can-only-go-down-when-taking-odd-powers-algorithmic}
For every odd integer $r\ge 3$, there exists an algorithm that takes as input a graph $G$ and a tree decomposition \hbox{$\mathcal{T} = (T, \{\Bag_t\}_{t\in V(T)})$} of $G$, and computes in time polynomial in $|V(G)|$
the graph $G^r$ and a tree decomposition \hbox{$\mathcal{T}' = (T, \{\Bag'_t\}_{t\in V(T)})$} of $G^r$ such that $\alpha_{G^r}(\mathcal{T}')\le \mu_G(\mathcal{T})$.
\end{corollary}

\subsubsection{Algorithmic consequences}

Let $d$ be a positive integer, let $G$ be a graph, let $\HH = \{H_j\}_{j\in J}$ be a finite family of connected non-null subgraphs of $G$, and let $J'\subseteq J$.
A \emph{distance-$d$ $\HH$-packing} in $G$ is a subfamily $\HH' = \{H_j\}_{j\in J'}$ of subgraphs from $\HH$ that are at pairwise distance at least $d$ in $G$.

\begin{observation}\label{observation}
Let $d$ be a positive integer, let $G$ be a graph, let $\HH = \{H_j\}_{j\in J}$ be a finite family of connected non-null subgraphs of $G$, and let $J'\subseteq J$.
Then, the corresponding subfamily $\HH' = \{H_i\}_{j\in J'}$ is a distance-$d$ $\HH$-packing in $G$ if and only if $\HH'$ is an independent $\HH$-packing in the graph $G^{d-1}$.
\end{observation}

\begin{proof}
Consider two distinct elements $i,j\in J'$.
It suffices to show that $\dist_G(V(H_i),V(H_{j})) \ge  d$ if and only if $\dist_{G^{d-1}}(V(H_i),V(H_{j})) \ge 2$, or, equivalently, that $\dist_G(V(H_i),V(H_{j})) \le d-1$ if and only if $\dist_{G^{d-1}}(V(H_i),V(H_{j})) \le 1$.

{Suppose} first that $\dist_G(V(H_i),V(H_{j})) \le d-1$ and let $P = (v_0,\ldots, v_r)$ be a path in $G$ from a vertex in $H_i$ to a vertex in $H_{j}$ such that $r\le d-1$.
If $r = 0$ then $\dist_{G^{d-1}}(V(H_i),V(H_{j}))= 0$.
If $r>0$, then in $G^{d-1}$, vertices $v_0$ and $v_{r}$ are adjacent, and hence $\dist_{G^{d-1}}(V(H_i),V(H_{j}))\le 1$.

{Suppose} now that $\dist_{G^{d-1}}(V(H_i),V(H_{j})) \le 1$.
If this distance is~$0$, then the graphs $H_i$ and $H_{j}$ have a vertex in common and hence $\dist_G(V(H_i),V(H_{j})) = 0\le d-1$.
If this distance is $1$, then there is a path $P$ in $G$ of length at most $d-1$ from a vertex in $H_i$ to a vertex in $H_{j}$.
Thus, $\dist_G(V(H_i),V(H_{j}))\le d-1$.
\end{proof}

For an integer $d$, the {input to} \textsc{Max Weight Distance-$d$ Packing} is a  graph $G$, a finite family $\HH = \{H_j\}_{j\in J}$ of connected non-null subgraphs of $G$, and a weight function $\wei\colon J\to \mathbb{Q}_+$ on the subgraphs in $\HH$.
The problem asks for a maximum-weight distance-$d$ $\HH$-packing in $G$.
By \cref{observation}, for every $J'\subseteq J$, the corresponding subfamily $\HH' = \{H_j\}_{j\in J'}$ is a distance-$d$ $\HH$-packing in $G$ if and only if $\HH'$ is an independent $\HH$-packing in the graph $G^{d-1}$.

\begin{theorem}\label{max-weight-distance-d-subgraph-packing-for-bounded-yw}
For every {positive integer $k$ and} even positive integer $d$, given a graph $G$ of induced matching treewidth {at most $k$}, a finite family $\HH = \{H_j\}_{j\in J}$ of connected non-null subgraphs of $G$, and a weight function $\wei\colon J\to \mathbb{Q}_+$ on the subgraphs in $\HH$, the \textsc{Max Weight Distance-$d$ Packing} problem is solvable in time polynomial in $|V(G)|$ and $|\cH|$.
\end{theorem}

\begin{sloppypar}
\begin{proof}
If $d = 2$, then the \textsc{Max Weight Distance-$d$ Packing} problem coincides with the \textsc{Max Weight Independent Packing} problem and the conclusion follows from \cref{max-weight-independent-subgraph-packing-for-yw}.
So we may assume that $d\ge 4$.

We compute the graphs $G^{d-1}$ and $(G^{d-1})^\circ[\HH]$; note that the latter graph is well-defined, since every graph in $\HH$ is a connected non-null subgraph of $G^{d-1}$.
By \cref{coro:ywGr}, $\tin(G^{d-1})\leq \yw(G)$, and by \cref{tin-of-G(H)-corollary}, $\tin((G^{d-1})^\circ[\HH])\leq \tin(G^{d-1})$.
Therefore, $\yw((G^{d-1})^\circ[\HH])\le \tin((G^{d-1})^\circ[\HH])\le \yw(G) = k$.
By \Cref{observation}, a subfamily $\HH'$ of $\cH$ is a distance-$d$ $\HH$-packing in $G$ if and only if $\HH'$ is an independent $\HH$-packing in the graph $G^{d-1}$, which is also equivalent to $\HH'$ being an independent set in the graph $(G^{d-1})^\circ[\HH]$.
Finally, using the fact that $\yw((G^{d-1})^\circ[\HH])\le k$, we call \cref{thm:yolovmwis} to solve the instance $((G^{d-1})^\circ[\HH], \wei)$ of \MWIS.
The overall running time is polynomial in $|V(G)|$ and $|\cH|$, as $k$ is a constant.
\end{proof}
\end{sloppypar}

Note that in the last step of the above proof, instead of \cref{thm:yolovmwis} we could have also used an algorithm for \MWIS{} {for} graphs with bounded $\tin$ (see~\cite{dallard2022firstpaper}).

We remark that unless $\P = \NP$, the result of \cref{max-weight-distance-d-subgraph-packing-for-bounded-yw} cannot be generalized to odd values of $d\ge 3$, since in this case the distance-$d$ variant of \textsc{Max Independent Set} is $\NP$-hard for chordal graphs  (see~\cite{DBLP:journals/jco/EtoGM14}), which have induced matching treewidth (and even tree-independence number) at most one.
Similarly, a result analogous to \cref{coro:ywGk} does not hold for even powers.
To this end, since the induced matching treewidth (resp.\ tree-independence number) of chordal graphs is at most one, it suffices to show
that the $k$-th powers of chordal graphs may have arbitrarily large induced matching treewidth (resp.\ tree-independence number). The idea of the construction below comes from Eto et al.~\cite{DBLP:journals/jco/EtoGM14}.

\begin{lemma}
    Let $r$ be an even positive integer. For any graph $H$, there exists a chordal graph~$G$ such that $G^r$ contains $H$ as an induced subgraph.
\end{lemma}

\begin{proof}
    Let $\widehat H$ be the graph obtained from $H$ by subdividing each edge exactly once and adding edges so that the newly introduced vertices form a clique. For convenience, we denote the vertices of $\widehat H$ \emph{not} resulting from edge subdivisions by $V(H)$. Note that the vertices of $V(\widehat H)\setminus V(H)$ correspond to the edges of $H$.
Then $V(H)$ is an independent set in $\widehat H$ and $V(\widehat H)\setminus V(H)$ is a clique in $\widehat H$.
Hence $\widehat H$ is a split graph and thus a chordal graph.
We obtain the graph $G$ from $\widehat H$ by appending to each vertex $v\in V(H)$ a path $P^v$ of length $(r-2)/2$ (which is an integer, since $r$ is even) such that one endpoint of $P^v$ is $v$ and all the other vertices of $P^v$ are new.
In particular, for any two different vertices $v,w\in V(H)$, the corresponding paths $P^v$ and $P^w$ are vertex-disjoint.
By construction, the graph $G$ is chordal.
For each vertex $v\in V(H)$, let us denote the two endpoints of $P^v$ by $v$ and $v'$ (with $v' = v$ if and only if $r = 2$), and let $X = \{v'\colon v\in V(H)\}$.
For arbitrary two distinct vertices $u,v\in V(H)$, it holds that $u'v'\in E(G^r)$ if and only if $\dist_G(u',v')\le r$, which happens if and only if $u$ and $v$ are adjacent in $H$.
Thus, the subgraph of $G^r$ induced by $X$ is isomorphic to $H$.
This shows that for every graph $H$ there exists a chordal graph $G$ such that $G^r$ contains an induced subgraph isomorphic to $H$, as claimed.
\end{proof}

Since there exist graphs with arbitrarily large induced matching treewidth and tree-independence number, we obtain the following corollaries.

\begin{corollary}
    Let $r$ be an even positive integer.
    There is no function $f$ satisfying $\tin(G^r) \leq f(\tin(G))$ for all graphs $G$.
\end{corollary}

\begin{corollary}
    Let $r$ be an even positive integer.
    There is no function $f$ satisfying $\yw(G^r) \leq f(\yw(G))$ for all graphs $G$.
\end{corollary}

\section{Max Weight Induced Subgraph for graphs with bounded tree-independence number}\label{sec:cmso}
We start this section by recalling some basic notions and properties related to \cmsotwo.

\subsection{\texorpdfstring{\cmsotwo}{CMSO\texttwoinferior}: basic notions and properties}

\paragraph{\cmsotwo logic on graphs.}
We assume that graphs are encoded as relational structures: each vertex and each edge is represented by a single variable (distinguishable by a unary predicate), and there is a single binary relation $\mathsf{inc}$ (\emph{incident}) binding each edge with its endvertices.

\msotwo (\emph{Monadic Second Order}) logic is a logic on graphs that allows us to use vertex variables, edge variables, {vertex set} variables, and edge set variables,
and, furthermore, to quantify over them.
An example of an \msotwo formula is the following expression checking if the chromatic number of a given graph is at most $r$ (for constant $r$).
Here, $S_i$'s are vertex set variables, $e$ is an edge variable, and $x,y$ are vertex variables.
\begin{align*}
\begin{split}
\exists_{S_1,S_2,\ldots,S_r} \; (\forall_{x} \; x \in S_1 \cup S_2 \ldots \cup S_r) \; \land
 \forall_{e} \forall_{x,y} \; (\mathsf{inc}(e,x) \land \mathsf{inc}(e,y) \land x \neq y) \to  \bigwedge_{i = 1}^r (x \notin S_i \lor y \notin S_i).	
 \end{split}
\end{align*}

For a positive integer $p$, by \cpmsotwo we mean the extension of \msotwo that allows us to use atomic formulae of the form $|S| \equiv a \bmod b$, where $S$ is a (vertex of edge) set variable, and $a,b \leq p$ are integers.
We define \cmsotwo (\emph{Counting Monadic Second Order} logic) to be  $\bigcup_{p >0}$\cpmsotwo.
The \emph{quantifier rank} of a formula is the maximum number of nested quantifiers in the formula (counting both first-order and second-order quantifiers).

\paragraph{Boundaried graphs and  \cmsotwo types.}
Let $\ell >0$ be an integer.
An \emph{$\ell$-boundaried graph} is a pair $(G,\iota)$, where $G$ is a graph and $\iota$ is a partial injective function from $V(G)$ to $[\ell]$. The domain of $\iota$, denoted by $\dom(\iota)$, is called the \emph{boundary} of $(G,\iota)$. For $v \in \dom(\iota)$, the value of $\iota(v)$ is called the \emph{label} of $v$.

Let us define two natural operations on boundaried graphs.
For an $\ell$-boundaried graph $(G,\iota)$ and $l \in [\ell]$,
the result of \emph{forgetting the label $l$} is the $\ell$-boundaried graph $(G,\iota_{\neg l})$,
here $\iota_{\neg l} \coloneqq \iota|_{\dom(\iota) \setminus \iota^{-1}(l)}$.
In other words, if the boundary of $(G,\iota)$ contains vertex $v$ with label $l$,
we remove the label from $v$, while keeping $v$ in the graph. Note that if {$l$ is not in the image of $\iota$}, then this operation does not do anything.

For two $\ell$-boundaried graphs $(G_1,\iota_1)$ and $(G_2,\iota_2)$,
the result of \emph{gluing} them is the $\ell$-boundaried graph $(G_1,\iota_1) \oplus_{[\ell]} (G_2,\iota_2)$ obtained from the disjoint union of $(G_1,\iota_1)$ and $(G_2,\iota_2)$ by identifying vertices with the same label.
Note that it might happen that some label  $l$ is used in, e.g., $\iota_1$ but not in $\iota_2$. Then the vertex with label $l$ in $(G_1,\iota_1) \oplus_{[\ell]} (G_2,\iota_2)$ is the copy of the vertex with label $l$ from $(G_1,\iota_1)$.

It is straightforward to observe that a graph has treewidth less than $\ell$ if and only if it can be constructed from $\ell$-boundaried graphs with at most two vertices by a sequence of forgetting and gluing operations (see, e.g.,~\cite{FominTV15}).

By \cmsotwo in $\ell$-boundaried graphs we mean \cmsotwo extended with $\ell$ unary predicates: for each $l \in [\ell]$ we have a predicate that selects a vertex with label $l$ (if such a vertex exists).
Let $(G,\iota)$ be an $\ell$-boundaried graph and $p$ and $q$ be positive integers.
The \emph{$(p,q)$-type} of $(G,\iota)$ is the set of all \cpmsotwo formulae $\phi$ of quantifier rank at most $q$ such that $(G,\iota) \models \phi$.
As such, the $(p,q)$-type of an $\ell$-boundaried graph is an infinite set.
However, the formulae in \cpmsotwo can be normalized so that every  formula
can be effectively transformed into an equivalent normalized formula of the same quantifier rank.
Furthermore, for every quantifier rank there are only finitely many pairwise non-equivalent normalized formulae.
This result is folklore (see, e.g.,~\cite[Lemma~6.1]{GroheK09}) and holds for general relational structures.
In the setting of $\ell$-boundaried graphs, the result is given by the following proposition.
The exact formulation we use here comes from Gartland et al.~\cite[Proposition~8]{gartland2020finding}.
\begin{proposition}\label{prop:mso-type}
For every triple of integers $\ell,p,q$, there exists a finite set $\msotypes^{\ell,p,q}$ and a function that assigns to every $\ell$-boundaried graph $(G,\iota)$ a \emph{type} $\msotype^{\ell,p,q}(G,\iota) \in \msotypes^{\ell,p,q}$ such that the following holds:
\begin{enumerate}
\item The types of isomorphic graphs are the same.
\item
  For every \cpmsotwo formula $\phi$ on $\ell$-boundaried graphs, whether $(G,\iota)$ satisfies $\phi$ depends only on the type $\msotype^{\ell,p,q}(G,\iota)$, where $q$ is the quantifier rank of $\phi$.
    More precisely,
  there exists a subset $\msotypes^{\ell,p,q}[\phi] \subseteq \msotypes^{\ell,p,q}$
  such that for every $\ell$-boundaried graph $(G,\iota)$ we have
  \[(G,\iota) \models \phi\qquad\textrm{if and only if}\qquad \msotype^{\ell,p,q}(G,\iota) \in \msotypes^{\ell,p,q}[\phi].\]
\item
 The types of ingredients determine the type of the result of the gluing operation.
More precisely, for every two types $\tau_1,\tau_2 \in \msotypes^{\ell,p,q}$
there exists a type $\tau_1 \oplus_{\ell,p,q} \tau_2$ such that
for every two $\ell$-boundaried graphs $(G_1,\iota_1)$, $(G_2,\iota_2)$,
if $\msotype^{\ell,p,q}(G_i,\iota_i) = \tau_i$ for $i=1,2$, then
\[\msotype^{\ell,p,q}((G_1,\iota_1) \oplus_{[\ell]} (G_2,\iota_2)) = \tau_1 \oplus_{\ell,p,q} \tau_2.\]
Also, the operation $\oplus_{\ell,p,q}$ is associative and commutative.
\item
 The type of the ingredient determines the type of the result of the forget label operation.
 More precisely,
 for every type $\tau \in \msotypes^{\ell,p,q}$ and $l \in [\ell]$
there exists a type $\tau_{\neg l}$ such that
for every $\ell$-boundaried graph $(G,\iota)$,
if $\msotype^{\ell,p,q}(G,\iota) = \tau$
and $(G,\iota_{\neg l})$ is the result of forgetting $l$ in $(G,\iota)$, then
\[\msotype^{\ell,p,q}(G,\iota_{\neg l}) = \tau_{\neg l}.\]
\end{enumerate}
\end{proposition}

Intuitively, for any $\ell$-boundaried graph $(G,\iota)$ and every \cpmsotwo formula $\phi$ with quantified rank at most $q$, whether $(G,\iota) \models \phi$ is fully determined by $\msotype^{\ell,p,q}(G,\iota)$.
Furthermore, if $\phi$ and $\ell$ are fixed, then the sets $\msotypes^{\ell,p,q}$ and $\msotypes^{\ell,p,q}[\phi]$ are of constant size.
Finally, it is known that there is an algorithm that, given $\ell$ and $\phi$, computes $\msotypes^{\ell,p,q}$ and $\msotypes^{\ell,p,q}[\phi]$~\cite{DBLP:journals/algorithmica/BoriePT92}.

\subsection{\texorpdfstring{\cmsotwo}{CMSO\texttwoinferior}-definable subgraphs of graphs with small tree-independence number}

Now we are ready to prove our main result concerning sparse induced subgraphs satisfying some \cmsotwo-definable property.
Here, ``sparse'' means having bounded clique number.

For positive integers $a$ and $b$, we denote by $\textsf{Ramsey}(a,b)$ the smallest positive integer $n$ such that every graph with at least $n$ vertices contains either a clique of size $a$ or an independent set of size $b$.
As proved by Ramsey~\cite{MR1576401}, these  numbers are well-defined.

\begin{theorem}\label{thm:cmsomain}
For every $k$, $r$, and a \cmsotwo formula $\psi$ there exists a positive integer $f(k,r,\psi)$ such that the following holds.
Let $G$ be a graph given along with a tree decomposition  $\cT=(T,\{X_t\}_{t \in V(T)})$ of independence number at most $k$,
and let $\wei\colon V(G) \to \Q_+$ be a weight function.
Then in time $f(k,r,\psi) \cdot |V(G)|^{\Oh(\textsf{Ramsey}(k+1,r+1))} \cdot |V(T)|$ we can find a set $F \subseteq V(G)$, such that
\begin{enumerate}
\item $G[F] \models \psi$,
\item $\omega(G[F]) \leq r$,
\item $F$ is of maximum weight subject to the conditions above,
\end{enumerate}
or conclude that no such set exists.
\end{theorem}

\begin{proof}
Let $\cT=(T,\{X_t\}_{t \in V(T)})$ be a tree decomposition of $G$ with independence number at most $k$.
Recall that without loss of generality we can assume that $\cT$ is nice, since otherwise we can transform the given tree decomposition into a nice one in time $\Oh(|V(G)|^2|V(T)|)$.
Let $F$ be a fixed (unknown) optimal solution. Note that by removing all vertices from $V(G) \setminus F$ from each bag of $\cT$, we obtain a tree decomposition of $G[F]$ with independence number at most $k$.
As $\omega(G[F]) \leq r$,  we conclude that every bag of $\cT$ contains fewer than $\textsf{Ramsey}(k+1,r+1)$ vertices of $F$.
In particular, the treewidth of $G[F]$ is at most $\ell\coloneqq \textsf{Ramsey}(k+1,r+1)$.
This is the only way in which we use that $\omega(G[F]) \leq r$.
Thus, in order to make sure that the obtained solution indeed satisfies the bound on the clique number,
instead of $\psi$ we will use the \cmsotwo formula $\phi \coloneqq \psi \land \phi_{\omega \leq r}$, where $\phi_{\omega \leq r}$ is the formula that bounds the clique number by $r$:
\[
    \forall_{v_1,\ldots,v_{r+1}} \lnot \left ( \bigwedge_{1 \leq i < j \leq r+1} \exists_e (\mathsf{inc}(e,v_i) \land \mathsf{inc}(e,v_j)) \right).
\]

Since the treewidth of the graph $G[F]$ is at most $\ell$, this graph can be constructed by a sequence of forgetting and gluing operations on $\ell$-boundaried graphs.

Consider the minimum integer $p$ such that $\phi$  belongs to \cpmsotwo, and let $q$ be the quantifier rank of $\phi$.
Let $\msotypes^{\ell,p,q}$ and $\msotypes^{\ell,p,q}[\phi]$ be as in \cref{prop:mso-type}.
Recall that these sets are of constant size (since $k,r$, and $\phi$ are fixed); we assume that they are hard-coded in our algorithm.
Similarly, we assume that the forgetting and gluing operations on types, i.e., functions
\[(\tau_1,\tau_2) \mapsto \tau_1 \oplus_{\ell,p,q} \tau_2\qquad\textrm{and}\qquad (\tau,l) \mapsto \tau_{\neg l},\]
are hard-coded. Finally, our algorithm has the list of types of all $\ell$-boundaried graphs with at most $\ell$ vertices.
As $\ell,p,q$ are fixed throughout the algorithm, we will simplify the notation by dropping subscripts and superscripts, i.e., we will simply write $\msotypes, \msotypes[\phi]$, and $\oplus$.

Our algorithm is a standard bottom-up dynamic programming on the nice tree decomposition $\mathcal{T}$ of $G$.
We will describe how to compute the weight of an optimal solution; adapting our approach to return the solution itself is a standard task.

We introduce the dynamic programming table $\Tab_{\cdot}[\cdot]$ indexed by nodes $t$ of $T$ and \emph{states}.
Each state is a triple $(S, \iota, \tau)$, where $S \subseteq X_t$ has at most $\ell$ vertices,
$\iota$ is an injective function from $S$ to $[\ell]$, and $\tau \in \msotypes$.
The value of $\Tab_t[S,\iota,\tau]$ will encode the maximum weight of a subset $F_t$ of $V(G_t)$ such that
\begin{enumerate}
\item $F_t \cap X_t = S$,
\item $\msotype((G[F_t],\iota))=\tau$.
\end{enumerate}
Intuitively, we think of $S$ as the intersection of a fixed optimal solution with $X_t$ and we look for a maximum-weight induced $\ell$-boundaried subgraph of $G_t$ with boundary $S$ and type $\tau$.
If there is no subset $F_t$ satisfying the conditions above, we will have $\Tab_t[S,\iota,\tau]=-\infty$.
Clearly, having computed all entries of $\Tab_\cdot[\cdot]$, the weight of the sought-for optimal solution is
\[
\max_{\tau} \Tab_{\root(T)}[\emptyset,\emptyset,\tau], \text{ over all } \tau \in \msotypes[\phi].
\]
Let us now describe how we compute the entries $\Tab_t[\cdot]$ for a node $t$ of $T$.
The exact way depends on the type of $x$.

\paragraph*{Case: $t$ is a leaf node.}
Note that in this case $G_t$ has no vertices, and thus the values of $\Tab_t[\cdot]$ can be extracted from the information that is hard-coded in the algorithm.

\paragraph*{Case: $t$ is an introduce node.}
Let $t'$ be the child of $t$ in $T$, and let $v$ be the unique vertex in $X_t \setminus X_{t'}$.
Clearly, if $v \notin S$, then $\Tab_t[S,\iota,\tau] = \Tab_{t'}[S,\iota,\tau]$.

So {let} $v \in S$, let $S' = S \setminus \{v\}$, and let $\iota' = \iota|_{S'}$.
Let $G_S$ be the subgraph of $X_t$ induced by $S$. As $|S| \leq \ell$, the value $\tau_S \coloneqq \msotype((G_S,\iota))$ is hard-coded in our algorithm.
Now we set
\[
\Tab_t[S,\iota,\tau] = \max_{\tau'} \Tab_{t'}[S',\iota',\tau'],
\]
where the maximum is taken over all $\tau'$ such that $\tau_S \oplus \tau' = \tau$.

\paragraph*{Case: $t$ is a forget node.}
Let $t'$ be the child of $t$ in $\cT$, and let $v$ be the unique vertex in $X_{t'} \setminus X_t$.
Intuitively, partial solutions at $t$ correspond to partial solutions at $t'$ with the label of $v$ (if any) forgotten.
Thus
\[
\Tab_t[S,\iota,\tau] = \max \left( \Tab_{t'}[S,\iota,\tau] , \max_{\iota',\tau'} \Tab_{t'}[S \cup \{v\},\iota',\tau'] \right),
\]
where the second term only appears if $|S| < \ell$.
The maximum in the second term is computed over all pairs $\iota',\tau'$, such that $\iota'|_{S}=\iota$ and $\tau'_{\neg \iota'(v)}=\tau$.

\paragraph*{Case: $t$ is a join node.}
Let $t',t''$ be the children of $t$.
Thus
\[
\Tab_t[S,\iota,\tau] = \max_{\tau',\tau''} \Tab_{t'}[S,\iota,\tau'] + \Tab_{t''}[S,\iota,\tau''] - \wei(S),
\]
where the maximum is taken over all $\tau',\tau''$ such that $\tau' \oplus \tau'' = \tau$. Note that we subtract the weight of $S$,
as it is counted both for $t'$ and for $t''$.

The correctness of the procedure above follows directly from the properties of nice tree decompositions and \cref{prop:mso-type}.

For each node $t$ of $T$, the number of states is at most $|X_t|^{\ell} \cdot \ell! \cdot |\msotypes|$.
Each entry of $\Tab_\cdot[\cdot]$ can be computed in time proportional to the number of states for each node (here we use that $k,r,\phi$ are fixed).
We conclude that the running time is $f(k,r,\psi) \cdot |V(G)|^{\Oh(\textsf{Ramsey}(k+1,r+1))}|V(T)|$ for some function $f$ depending on $k,r$, and $\psi$.
\end{proof}

We remark that one can easily obtain { a variant of \cref{thm:cmsomain} in which a tree decomposition with small independence number is not needed in the input.
Indeed, } combining \cref{thm:cmsomain} {with \cref{thm:approximate-dallard}} yields \cref{thm:cmsomain-intro} (with $g(k,r) =\Oh(\textsf{Ramsey}(8k+1,r+1))$).

\begin{theorem}
\label{thm:cmsomain-intro}
For every $k$, $r$, and a \cmsotwo formula $\psi$ there exist positive integers $f(k,r,\psi)$ and $g(k,r)$ such that the following holds.
Let $G$ be a graph with $\tin(G) \leq k$, and let $\wei\colon V(G) \to \Q_+$ be a weight function.
Then in time $f(k,r,\psi) \cdot |V(G)|^{g(k,r)}$ we can find a set $F \subseteq V(G)$, such that
\begin{enumerate}
\item $G[F] \models \psi$,
\item $\omega(G[F]) \leq r$,
\item $F$ is of maximum weight subject to the conditions above,
\end{enumerate}
or conclude that no such set exists.
\end{theorem}

\subsection{A generalization}

\cref{thm:cmsomain-intro} can be slightly generalized so that we can maximize the weight of some carefully chosen subset of the solution.
The formula $\psi$ uses this set as a free variable.

\begin{theorem}\label{thm:cmsomainX}
For every $k$, $r$, and a \cmsotwo formula $\psi$ with one free vertex-set variable, there exist positive integers $f(k,r,\psi)$ and $g(k,r)$ such that the following holds.
Let $G$ be a graph with $\tin(G) \leq k$, and let $\wei\colon V(G) \to \Q_+$ be a weight function.
Then in time $f(k,r,\psi) \cdot |V(G)|^{g(k,r)}$ we can find sets $X \subseteq F \subseteq V(G)$, such that:
\begin{enumerate}
\item $(G[F],X) \models \psi$,
\item $\omega(G[F]) \leq r$,
\item $X$ is of maximum possible weight, subject to conditions above,
\end{enumerate}
or conclude that no such set exists.
\end{theorem}
\begin{proof}[Proof sketch] The proof follows by a simple trick used by Gartland et al.~\cite[Theorem 38]{gartland2020finding}.
We just sketch the argument here and refer the reader to~\cite{gartland2020finding} for more details.

For a graph $G$ and a subset $M \subseteq V(G)$, by $\forked{G}$ we denote the \emph{forked version of} $(G,M)$, which is obtained as follows.
We start the construction with a copy of $G$.
For every vertex $v$ of $G$ we introduce to $\forked{G}^M$ three vertices of degree 1, each adjacent only to $v$.
For every vertex $v$ of $M$ we introduce to $\forked{G}^M$ a two-edge path, whose one endvertex is identified with $v$.
If $G$ is equipped with weights $\wei$, we set the weights $\forked{\wei}$ of all vertices of $\forked{G}^M$ to 0, with the exception that the value of $\wei(v)$ for $v \in M$ is transferred to the other endvertex of the two-edge path attached to~$v$.

Note that if $F'$ is a forked version of some $(F,M)$, then $F$ and $M$ can be uniquely decoded from $F'$.
Indeed, $F$ is the subgraph of $F'$ induced by vertices of degree at least four, and vertices of $M$ are those vertices of $F$ that are adjacent (in $F'$) to a vertex of degree 2.
Furthermore, $\omega(F') \leq \max(\omega(F),2)$.

We apply the algorithm of {\cref{thm:approximate-dallard} to compute in time $|V(G)|^{\Oh(k)}$} a tree decomposition $\cT=(T,\{X_t\}_{t \in V(T)})$ of $G$ with independence number at most $8k$.
Moreover, observe that, given the tree decomposition $\cT$, we can easily obtain a tree decomposition $\forked{\cT}=(\forked{T},\{\forked{X}_t\}_{t \in V(\forked{T})})$ of $\forked{G}^{V(G)}$ with $\alpha(\forked{\cT}) \leq \alpha(\cT)$.
The idea is to call the algorithm from \cref{thm:cmsomain} on $\forked{G}^{V(G)}$ and a formula $\forked{\psi}$ that for a graph $F'$ ensures that
\begin{itemize}
\item if $F'$ is a forked version of some $(F,M)$, then $F' \models \forked{\psi}$ if and only if $(F,M) \models \psi$,
\item otherwise $F' \not\models \forked{\psi}$.
\end{itemize}
Such a formula can be easily constructed in \cmsotwo~\cite{gartland2020finding}.
We observe that the solution found by the algorithm satisfies the conditions stated in the theorem; again, see~\cite[Theorem 38]{gartland2020finding} for details.
\end{proof}

\subsection{Algorithmic applications and limitations}
Using the framework of \cref{thm:cmsomain-intro} we can solve several well-studied problems {for graphs of bounded tree-independence number}:
\begin{sloppypar}
\begin{itemize}
\item \textsc{Max Weight Independent Set},
\item \textsc{Max Weight Induced Forest}, which is equivalent by complementation to \textsc{Min Weight Feedback Vertex Set},
\item \textsc{Max Weight Induced Bipartite Subgraph}, which is equivalent by complementation to \textsc{Min Weight Odd Cycle Transversal},
\item \textsc{Max Weight Induced Odd Cactus}, which is equivalent by complementation to \textsc{Min Weight Even Cycle Transversal} (see e.g.~\cite{DBLP:conf/iwpec/BergougnouxBBK20}),
\item \textsc{Max Weight Induced Planar Subgraph}, which is equivalent to \textsc{Planarization} (see e.g.~\cite{DBLP:journals/dam/Pilipczuk17}),
\item {\textsc{Defective Coloring}, where we ask for a vertex coloring using  $\chi_{\textrm{d}}$ colors in which each color class induced a subgraph of maximum degree bounded by $\Delta^*$, where both $\chi_{\textrm{d}}$ and $\Delta^*$ are fixed (see, e.g.,~\cite{zbMATH07538442}), and}
\item {\textsc{Max $k$-Dependent Set}, that is, maximum set of vertices inducing a subgraph of maximum degree at most $k$ (see, e.g.,~\cite{zbMATH05818768}),  which is equivalent by complementation to \textsc{Bounded-Degree Vertex Deletion}.}
\end{itemize}
\end{sloppypar}

{In particular, these results generalize Theorems 23 and 25 from~\cite{zbMATH07538442}, giving polynomial-time algorithms for \textsc{Defective Coloring} for fixed $\chi_{\textrm{d}}$ and $\Delta^*$ in graphs of bounded treewidth and chordal graphs, respectively.
They also solve an open problem from~\cite{zbMATH05818768} regarding the complexity of \textsc{Max $k$-Dependent Set} on interval graphs (see \cite[Table 2]{zbMATH05818768}).}

Using the extension from \cref{thm:cmsomainX} we can solve some other problems, for example:
\begin{itemize}
\item find the maximum number of pairwise independent induced cycles.
\end{itemize}
Note that the difference is that we maximize the \emph{number} of cycles, and not their \emph{total size} (or weight), as we could do using \cref{thm:cmsomain-intro}.
On other other hand, the number of induced cycles in a graph might be very large, so this problem cannot be solved using \cref{max-weight-independent-subgraph-packing-for-yw} (or \cref{max-weight-distance-d-subgraph-packing-for-bounded-yw}).
However, we can solve our problem by calling \cref{thm:cmsomainX} with the following parameters:
\begin{itemize}
    \item $r=3$,
    \item $\psi$ is the formula saying that $G[F]$ is 2-regular (i.e., each component of $G[F]$ is a cycle) and each vertex from $X$ must be in a distinct connected component of $G[F]$, and
    \item the weight of each vertex is 1.
\end{itemize}

In conclusion, let us point out that in \cref{thm:cmsomain-intro,thm:cmsomain,thm:cmsomainX} we need to assume that the clique number of the solution is bounded by a constant.
Indeed, the property of being a clique can be easily expressed by a \cmsotwo formula as follows (here $x,y$ are vertex variables and $e$ is an edge variable):
\[
    \phi_{\textrm{clique}} \coloneqq \forall_{x,y} \left( x \neq y \to \exists_e \left(\textsf{inc}(e,x) \land \textsf{inc}(e,y)\right)  \right).
\]
However, \textsc{Max Clique} is \NP-hard on graphs with independence number 2, and thus of tree-independence number at most 2. Indeed, this follows from the fact that \textsc{Max Independent Set} is \NP-hard on triangle-free graphs~\cite{Po74}, by complementing the input graph.

\begin{sloppypar}
\section{Complexity of computing induced matching treewidth and bounds}\label{sec:structural}

In this section, we collect some basic complexity and algorithmic observations about induced matching treewidth, establish its close relation to treewidth for graphs with bounded degree, and give constructions of two families of graphs with unbounded degree and unbounded induced matching treewidth.

\subsection{Remarks on complexity of computing induced matching treewidth}

\Cref{lem:Yolov} together with known hardness results on the tree-independence number from~\cite{dallard2022firstpaper,dallard2022computing} and on the independence number from~\cite{MR2403018} imply the following hardness results for induced matching treewidth.

\begin{theorem}\label{thm:hardness}
The following lower bounds hold.
\begin{enumerate}
    \item For every constant $k\geq 4$, it is \NP-complete to decide whether $\yw(G)\leq k$ for a graph~$G$.
    \item For every $\varepsilon>0$, there is no polynomial-time algorithm for approximating induced matching treewidth of an $n$-vertex graph to within a factor of $n^{1-\varepsilon}$ unless $\P = \NP$.
    \item There is no constant-factor \FPT-approximation algorithm for induced matching treewidth, unless $\FPT = \Wone$.
    \item For any computable functions $f,g$, there is no $g(k)$-approximation algorithm for computing induced matching treewidth in time $f(k) \cdot n^{o(k)}$, unless the Gap-ETH fails.\footnote{Gap-ETH states that for some constant $\epsilon > 0$, distinguishing between a satisfiable 3-\textsf{SAT} formula and one that is not even $(1-\epsilon)$-satisfiable requires exponential time (see~\cite{DBLP:conf/icalp/ManurangsiR17,DBLP:journals/eccc/Dinur16}).}
\end{enumerate}
\end{theorem}

\begin{proof}
Fix an integer $k\ge 4$.
Clearly, the problem of testing if $\yw(G)\leq k$ for a given graph $G$ belongs to \NP.
To establish \NP-hardness, we make a reduction from the problem of testing if a given graph $G$ satisfies $\tin(G)\le k$.
As shown by Dallard et al.~\cite{dallard2022computing}, this problem is \NP-complete.
By \Cref{lem:Yolov}, $\tin(G) = \yw(G\circ K_1)$ and hence $\tin(G)\le k$ if and only if $\yw(G\circ K_1)\le k$.
Since computing the corona of $G$ can be done in polynomial time, hardness follows.

To obtain the inapproximability result, let us recall the original reduction establishing \NP-hardness of the problem of computing the tree-independence number due to Dallard et al.~\cite{dallard2022firstpaper}.
The reduction was from the independent set problem and was based on the fact that for any graph $G$, if $G'$ is the graph obtained from two disjoint copies of $G$ by adding all possible edges between them, then $\alpha(G) = \tin(G')$.
Combining this result with~\Cref{lem:Yolov}, we obtain that $\alpha(G) = \tin(G') = \yw(G' \circ K_1)$.
Given an $n$-vertex graph $G$, the graph $G'\circ K_1$ can be computed in polynomial time from $G$ and has $4n$ vertices.
{Now,} approximately computing the induced matching treewidth of the transformed graph would imply an equally good approximation to the independence number of $G$.

Thus statements 2., 3., 4.~{follow} from analogous hardness results for approximating \textsc{Max Independent Set}, shown by, respectively,
Zuckerman~\cite{MR2403018}, Lin~\cite{DBLP:conf/stoc/Lin21}, and Chalersmook et al.~\cite{DBLP:journals/siamcomp/ChalermsookCKLM20}.
\end{proof}

Let us remark that \cref{thm:hardness} {does not leave} much space for improvements {to} \cref{thm:approximate-yolov}.
First, it implies that {we} cannot \emph{exactly} compute induce matching treewidth, even for small values (at least 4).
Second, we cannot hope for a nontrivial approximation in polynomial time.
Finally, the running time of the approximation algorithm cannot be improved to $f(k) \cdot n^{\Oh(1)}$ and even $f(k) \cdot n^{o(k)}$, under standard complexity assumptions.

\medskip
\Cref{thm:hardness} establishes \NP-completeness of recognizing graphs with induced matching width at most $k$, for every integer $k\ge 4$.
The complexity of recognizing graphs with induced matching treewidth at most $k$ is open for $k\in \{2,3\}$, but solvable in polynomial time for $k\le 1$, as we discuss next.
Before doing that, let us show a nice property of induced matching treewidth that is also satisfied by tree-independence number (see~\cite{dallard2022firstpaper}) and {by \emph{sim-width}, a graph parameter introduced in 2017 by Kang, Kwon, Str\o{}mme, and Telle (see~\cite{MR3721445}):} the parameter is monotone under induced minors.
Given two graphs $G$ and $H$, we say that $H$ is an \emph{induced minor} of $G$ if a graph isomorphic to $H$ can be obtained from $G$ by a sequence of vertex deletions and edge contractions.

\begin{proposition}\label{prop:induced-minors}
Let $G$ be a graph and let $H$ be an induced minor of $G$.
Then $\yw(H)\le \yw(G)$.
\end{proposition}

\begin{proof}
We show how to modify any tree decomposition $\mathcal{T} = (T, \{X_t\}_{t\in V(T)})$ of $G$ into a tree decomposition $\mathcal{T}'$ of $H$ such that the induced matching treewidths of the two decompositions are related as follows: $\mu_H(\mathcal{T}')\le \mu_G(\mathcal{T})$.
This will imply the desired inequality  $\yw(H)\le \yw(G)$.

Clearly, it suffices to consider the case when $H$ is obtained from $G$ either by a single vertex deletion or by a single edge contraction.

Consider first the case when $H$ is the graph obtained from $G$ by deleting some vertex, say $v\in V(G)$.
Then deleting $v$ from all the bags of $\mathcal{T}$ that contain $v$ results in a tree decomposition $\mathcal{T}'$ of $H$.
Furthermore, since every induced matching in $H$ is also an induced matching in $G$, it holds that $\mu_G(\mathcal{T})\ge \mu_H(\mathcal{T}')$, as claimed.

Next, consider the case when $H$ is the graph obtained from $G$ by contracting an edge $uv\in E(G)$ into a single vertex $z$.
In this case, it is not difficult to see that replacing each bag $X_t$ of $\mathcal{T}$ that contains either $u$ or $v$ (or both) with the bag $X_t'\coloneqq (X_t\setminus\{u,v\})\cup\{z\}$ results in a tree decomposition $\mathcal{T}'$ of $H$.
We show that $\mu_G(\mathcal{T})\ge \mu_H(\mathcal{T}')$ by verifying that for every node $t\in V(T)$ and every induced matching $M'$ in $H$ that touches the bag $X_t'$, there exists an induced matching $M$ in $G$ that touches the bag $X_t$ and such that $|M| = |M'|$.
This will indeed suffice, since taking $M'$ to be a maximum induced matching in $H'$ that touches some bag, say $X_t'$, of $\mathcal{T}'$, then the corresponding induced matching $M$ in $G$ that touches the bag $X_t$ will imply that
$\mu_G(\mathcal{T})\ge |M| = |M'| = \mu_H(\mathcal{T}')$.

So, let $t\in V(T)$ be a node of $T$ and let $M'$ be an induced matching in $H$ that touches the bag $X_t'$.
We consider two cases depending on whether the vertex $z$ is an endpoint of an edge in $M'$ or not.

{Suppose} first that the vertex $z$ is not an endpoint of any edge in $M'$.
Since the graph $H-z$ is an induced subgraph of $G$, the matching $M'$ is also an induced matching in $G$.
Furthermore, since the matching $M'$ touches the bag $X_t'$ in vertices different from $z$, the fact that $X_t'\setminus\{z\}\subseteq X_t$ implies that $M'$ also touches the bag $X_t$.
Hence we can take $M = M'$ in this case.

{Suppose} next that the vertex $z$ is an endpoint of an edge in $M'$.
In this case there is a unique edge $e\in M'$ such that $z$ is an endpoint of $e$.
We consider two further subcases depending on whether $z$ belongs to the bag $X_t'$ or not.
{Suppose} first that $z\in X_t'$.
In this case, $X_t'=(X_t\setminus\{u,v\})\cup\{z\}$ and the set $M \coloneqq (M'\setminus \{e\})\cup\{uv\}$ is an induced matching in $G$ such that $|M| = |M'|$.
Furthermore, since the matching $M\setminus \{uv\} = M'\setminus \{e\}$ touches the set $X_t\setminus \{u,v\} = X_t'\setminus \{z\}$ and the edge $uv$ touches $X_t$ (since $z\in X_t'$), we infer that $M$ touches $X_t$.

Finally, {suppose} that $z\not\in X_t'$.
Then $X_t' = X_t$.
Let $y$ be the endpoint of $e$ other than $z$.
Since the edge $e$ touches the bag $X_t$ and $z\not\in X_t'$, we have $y\in X_t$.
Moreover, $y$ is a vertex of $G$ such that $uy\in E(G)$ or $vy\in E(G)$.
We may assume without loss of generality that $uy\in E(G)$.
Then, the set $M\coloneqq (M'\setminus \{e\})\cup\{uy\}$ is an induced matching in $G$ such that $|M| = |M'|$.
Furthermore, since the matching $M\setminus \{uy\} = M'\setminus \{e\}$ touches the set $X_t= X_t'$ and the edge $uy$ touches $X_t$, we infer that $M$ touches $X_t$.
This completes the proof.
\end{proof}

\Cref{prop:induced-minors} implies the following.
For an integer $k\ge 0$, let $\mathcal{F}_k$ denote the {family} of all \emph{forbidden induced minors} for the class of graphs with induced matching treewidth at most $k$, that is, $\mathcal{F}_k$ consists precisely of the graphs $G$ such that $\yw(G)>k$ but every proper induced minor $H$ of $G$ satisfies $\yw(H)\le k$.
Then, by
\Cref{prop:induced-minors}, a graph $G$ has induced matching treewidth at most $k$ if and only if $G$ does not have any graph from $\mathcal{F}_k$ as an induced minor.

It is not difficult to see that the only graphs with induced matching treewidth equal to $0$ are exactly the edgeless graphs; hence $\mathcal{F}_0 = \{K_2\}$.
To characterize graphs with induced matching treewidth at most $1$, we will need the following result due to Scheidweiler and Wiederrecht~\cite[Theorem 3.6]{MR3804753}.
We denote by $C_k$ the $k$-vertex cycle graph.

\begin{theorem}\label{thm:SW}
There exists a finite set $\mathcal{F}$ of graphs such that for every graph $G$, the graph $L^2(G)$ is chordal if and only if $G$ does not contain any graph from the set $\mathcal{F}\cup\{C_k:k\ge 6\}$ as an induced subgraph.
\end{theorem}

We now have everything ready to prove the following.

\begin{proposition}
The class of graphs with induced matching treewidth at most $1$ is characterized by finitely many forbidden induced minors and admits a polynomial-time recognition algorithm.
\end{proposition}

\begin{proof}
By \Cref{lem:Yolov}, for every graph $G$, it holds that $\yw(G) = \tin(L^2(G))$.
Recall also that a graph $G$ is chordal if and only if $\tin(G)\le 1$.
It follows that a graph $G$ has induced matching treewidth at most $1$ if and only if the square of its line graph is chordal.
By \Cref{thm:SW}, the class of graphs with this property is characterized by a family of forbidden induced subgraphs consisting of all cycles of length at least $6$ and a family $\mathcal{F}$ of finitely many other graphs.
Since each forbidden induced minor is also a forbidden induced subgraph, the family of forbidden induced minors is a subfamily of $\mathcal{F}\cup\{C_k:k\ge 6\}$.
Using also the fact that the $6$-cycle is an induced minor of any longer cycle, we infer that the set $\mathcal{F}_1$ of forbidden induced minors for the class of graphs with induced matching treewidth at most~$1$ consists of the cycle $C_6$ and a subfamily of $\mathcal{F}$; in particular, it is finite.

Given a graph $G$, testing if the square of its line graph is chordal can be done in polynomial time (using, e.g.,~\cite{doi:10.1137/0205021}).
Therefore, graphs with induced matching treewidth at most $1$ can be recognized in polynomial time.
\end{proof}

\subsection{Induced matching treewidth in graphs with bounded degree}

Recall that the tree-independence number of a graph is an
upper bound on its induced matching treewidth, hence, any graph class with unbounded induced matching treewidth is necessarily of unbounded tree-independence number.
On the other hand, large tree-independence number does not necessarily imply large induced matching treewidth, as shown, for example, by the family of complete bipartite graphs.
The balanced complete bipartite graphs $K_{n,n}$ satisfy $\tin(K_{n,n}) = n$ (see~\cite{dallard2022firstpaper}) and $\yw(K_{n,n}) = 1$, as witnessed by the trivial tree decomposition having a unique bag equal to the entire vertex set.

\begin{sloppypar}
{
Interestingly, in the absence of complete bipartite graphs as subgraphs, bounded induced matching treewidth implies bounded tree-independence number and even bounded treewidth.
This follows from a result of Brettell et al.~\cite[Proposition 18]{DBLP:journals/corr/abs-2308-05817} who proved a related result for sim-width.
Brettell et al.~showed that for any two integers $k,t \in \mathbb{N}$ there exists an integer $r$ such that each graph with sim-width at most $k$ and no subgraph isomorphic to $K_{t,t}$ has treewidth at most~$r$.
It is known that bounded induced matching treewidth implies bounded sim-width.
This is explicitly shown in~\cite[Theorem 24]{DBLP:journals/corr/abs-2302-10643} (see also~\cite{MR4657707}) but in fact the argument boils down to the proof of~\cite[Proposition 3.1]{MR3721445}.
The formal statement of the result for induced matching treewidth is as follows.

\begin{theorem}\label{thm:tw-biclique}
For any two integers $k,t \in \mathbb{N}$ there exists an integer $r$ such that each graph with induced matching treewidth at most $k$ and no subgraph isomorphic to $K_{t,t}$ has treewidth at most~$r$.
\end{theorem}

Note that graphs with maximum degree less than $t$ have no subgraph isomorphic to $K_{t,t}$.
Hence, \Cref{thm:tw-biclique} implies that any class of graphs with bounded degree and unbounded treewidth has unbounded induced matching treewidth.
This includes, for example, the classes of (subdivided) square or hexagonal grids (also known as \emph{walls}) and their line graphs.

An equivalent formulation is that in classes of graphs with bounded induced matching treewidth, bounded degree implies bounded treewidth.
This implication (that bounded degree implies bounded treewidth) was recently established by Abrishami, Chudnovsky, and Vu\v{s}kovi\'{c}~\cite{ABRISHAMI2022144} for the class of even-hole-free graphs, and by Korhonen~\cite{KORHONEN2023206} for any class of graphs excluding some fixed grid (or equivalently, some fixed planar graph) as an induced minor, as well as for any class of graphs excluding all sufficiently large walls and their line graphs as induced subgraphs.

The proof of \Cref{thm:tw-biclique} via~\cite{DBLP:journals/corr/abs-2308-05817} is based on Ramsey-type arguments and the binding function is quite large (in particular, it is exponential in $t$).
We now give a simple argument showing that for graphs with bounded degree, the resulting bound on treewidth is linear in $k$ and polynomial in the maximum degree.
We first show that} for graphs with bounded degree, induced matching treewidth is linearly related to the tree-independence number.
As usual, for a graph $G$, we denote by $\Delta(G)$ the maximum degree of a vertex  in $G$.
\end{sloppypar}

\begin{theorem}\label{thm:delta}
For every graph $G$ with at least one edge, it holds that $\tin(G) \leq 2\yw(G) \cdot \Delta(G)^2$.
\end{theorem}

\begin{proof}
    Let $G$ be a graph with at least one edge.
    If $G$ has an isolated vertex, we can delete it without changing the value of any of the three parameters $\tin(G)$, $\yw(G)$, and $\Delta(G)$, since they are all strictly positive.
    We may therefore assume that $G$ has no isolated vertices.
    To shorten the notation we will denote $\Delta \coloneqq \Delta(G)$.
    Let $\cT = (T,\{X_t\}_{t \in V(T)})$ be a tree decomposition of $G$ witnessing $\yw(G)$.
    For all $t\in V(T)$, let
    $X'_t = N[X_t]$ and let $\cT' = (T,\{X'_t\}_{t \in V(T)})$.
    It is straightforward to verify that $\cT'$ is a tree decomposition of $G$.
    We want to argue that the independence number of each bag of $\cT'$ is at most $2\yw(G) \cdot \Delta^2$.

    Consider a bag $X'_t$ of $\cT'$.
    It consists of $X_t$ and $Y = N(X_t)$.
    Let $I$ be a largest independent set in $G[X'_t]$ and let $k$ be its size.
    We iteratively construct an induced matching $M$ in $G$ whose every edge intersects $X_t$ and $|M| \geq k/(2\Delta^2)$.
    As $|M| \leq \yw(G)$, this will give $|I|=k \leq \yw(G) \cdot 2\Delta^2$.

    At the beginning, the matching $M$ is empty.
    In each step of the procedure we add one edge to $M$ and discard some set of vertices.
    The procedure stops when all the vertices in $I$ are discarded.
    At each step of the procedure we want to preserve the following properties:
    \begin{enumerate}
        \item Each $v \in I$ that is not discarded has a non-discarded neighbor in $X'_t$.
        \item Each $v \in I \cap Y$ that is not discarded has a non-discarded neighbor in $X_t$.
    \end{enumerate}
    Clearly the above properties are satisfied at the beginning, before any vertex is discarded, as the graph has no isolated vertices and every vertex from $I\cap Y$ has a neighbor in $X_t$.

    Now pick any vertex $v$ in $I$ that was not discarded yet.
    By properties 1.~and 2.~$v$ has a non-discarded neighbor $w$, and if $v \in Y$, then $w$ can be chosen to be in $X_t$.
    We then include the edge $vw$ in the induced matching $M$.

    Next, we discard the following vertices from $X'_t$:
    \begin{itemize}
        \item $v$ and $w$,
        \item every non-discarded neighbor of $v$ or $w$, and
        \item every non-discarded vertex in $I$ that is at distance 2 from $\{v,w\}$ in $G[X'_t]$.
    \end{itemize}
    Notice that since we removed $N[\{u,v\}] \cap X'_t$, the edges chosen in every subsequent iteration will indeed form an induced matching.
    In each iteration we discard at most $2\Delta^2$ vertices from $I$, so,
    provided that properties 1. and 2. are preserved, the number of iterations we perform to build an induced matching $M$ of the desired size is at most $k/(2\Delta^2)$.
    Note also that every edge of $M$ intersects~$X_t$.

    Finally, we prove that at each step of the procedure properties 1. and 2. hold.
    To this end, it suffices to show that for each $u \in I$ that is not discarded, its entire neighborhood in $X'_t$ is not discarded.
    Suppose for a contradiction that there exists some vertex $u' \in N(u) \cap X'_t$ that is discarded.
    Clearly $u' \notin I$, so $u'$ must be equal or adjacent to $v$ or $w$.
    However, this means that $u$ is at distance at most 2 from $\{v,w\}$ in $X'_t$, so $u$ got discarded, a contradiction.
\end{proof}

\begin{corollary}\label{cor:treewidth}
For every graph $G$ with at least one edge, it holds that
\[\tw(G) \leq 2\yw(G) \cdot \Delta(G)^2(\Delta(G)+1)\,.\]
\end{corollary}

\begin{proof}
Let $G$ be a graph with at least one edge, let $\Delta$ be its maximum degree and let $k$ be its induced matching treewidth.
By \Cref{thm:delta}, $G$ admits a tree decomposition in which each bag induces a subgraph with independence number at most $2k\Delta^2$.
Since the maximum degree of any such subgraph is at most $\Delta$, its chromatic number is at most $\Delta+1$.
Consequently, every bag has at most $2k\Delta^2(\Delta+1)$ vertices.
\end{proof}

\subsection{Two graph families with unbounded induced matching treewidth}

We conclude this section by identifying two graph families with unbounded maximum degree and unbounded induced matching treewidth.
For the proofs, we need the following well-known property of tree decompositions (see, e.g.,~\cite{dallard2022firstpaper}).

\begin{lemma}\label{lem:nvinbag}
For every graph $G$ and every tree decomposition $\cT$ of $G$, there exists a vertex $v\in V$ such that $N[v]$ is contained in some bag of $\cT$.
\end{lemma}

For a positive integer $n$, the \emph{$n$-dimensional hypercube graph} is the graph $Q_n$ with vertex set $\{0,1\}^n$ in which two vertices $x = (x_1,\ldots, x_n)$ and $y= (y_1,\ldots, y_n)$ are adjacent if and only if there exists a unique $i\in \{1,\ldots, n\}$ such that $x_i\neq y_i$.

\begin{proposition}
The $n$-dimensional hypercube graph $Q_n$ has $\yw(Q_n)\geq \lfloor\frac{n}{2}\rfloor$.
\end{proposition}

\begin{proof}
Consider an arbitrary tree decomposition $\cT = (T,\{X_t\}_{t\in V(T)})$ of $Q_n$.
By \cref{lem:nvinbag}, there exists a bag $X_t$ that contains $N[v]$ for some $v\in V(Q_n)=\{0,1\}^n$.
By the symmetry of $Q_n$, we may assume without loss of generality that $v = (0,\ldots,0)$.

Let $v_i$ denote the vertex whose only non-zero entry is at the $i$-th position and let $w_{i,j}$ denote the vertex whose only non-zero entries are the $i$-th and $j$-th position.

An induced matching $M$ touching $N[v]$ (and thus touching the bag) of size $\lfloor\frac{n}{2}\rfloor$ is as follows
\[M=\left\{v_{2i}w_{2i,2i-1} \large\mid i\in \left\lfloor\frac{n}{2}\right\rfloor\right\}.\]
This implies that $\mu(\cT)$ is at least $\lfloor\frac{n}{2}\rfloor$.
Since $\cT$ was arbitrary, we infer that $\yw(Q_n)\geq \lfloor\frac{n}{2}\rfloor$.
\end{proof}

Our final result gives a simple construction of a family of $P_4$-free graphs with unbounded induced matching number.
Recall that for a positive integer $n$, we denote by $nK_2$ the disjoint union of $n$ copies of $K_2$.

\begin{proposition}\label{prop:matching-biclique}
Let $n$ be a positive integer and let $G_n$ denote the graph obtained from two disjoint copies of the graph $nK_2$ by adding all possible edges between them.
Then $\yw(G_n)\geq n$.
\end{proposition}

\begin{proof}
Let $\cT = (T, \{X_t\}_{t\in V(T)})$ be an arbitrary tree decomposition of $G_n$.
By \cref{lem:nvinbag}, there exists some bag $X_t$ containing $N[v]$ for some $v\in V(G_n)$. By definition of $G_n$, the bag $X_t$ contains all the vertices of one of the copies of $nK_2$, which form an induced matching of size~$n$.
\end{proof}
\end{sloppypar}

\section{Conclusion}\label{sec:conclusion}
The obvious direction for further research is to show \cref{conjecture}.
A somewhat easier, but still very interesting problem, would be to show the following result, which is a significant common generalization of both \cref{thm:fvs} and \cref{cor:ptas}.

\begin{conjecture}\label{con:treewidth}
Let $k,r \in \mathbb{N}$ be fixed.
Let $G$ be a graph with induced matching treewidth at most $k$, equipped with a weight function $\wei\colon V(G) \to \Q_+$.
In polynomial time we can find a maximum-weight induced subgraph of $G$ with treewidth at most $r$.
\end{conjecture}

We believe that the right way of approaching \cref{con:treewidth} is to prove an analogue of \cref{prop:yolov-misets} and \cref{lem:ff}:
 for each node of a tree decomposition of bounded induced matching treewidth, there is only a polynomial number of states of the natural dynamic programming algorithm (developed for bounded-treewidth graphs~\cite{DBLP:journals/jal/BodlaenderK96}) that represent an inclusion-maximal induced subgraph of treewidth at most $r$. However, the space of states in these algorithms is quite complicated (much more complicated than when finding a largest induced forest) and it is far from obvious how to re-interpret it in the current setting.

\medskip
Another interesting question is to try to find a ``better parameter'' than induced matching width.
If we think of a parameter useful for solving \MWIS, it is natural to expect that for a simple graph such that $K_{n,n} \odot K_1$, i.e., a biclique with additional pending edge added to each vertex, such a parameter would be bounded. However, by \cref{lem:Yolov}, we have ${\yw(K_{n,n} \odot K_1)} = \tin(K_{n,n}) = n$.

Recently, Bergougnoux et al.~\cite{MR4657707} introduced \emph{one-sided mim-width} (or \emph{o-mim-width}), which can be seen as a unification of tree-independence number and \emph{mim-width}, which is another well-established graph width parameter useful in solving \MWIS and many other problems~\cite{DBLP:conf/soda/BergougnouxDJ23}.
{Just like the definition of induced matching treewidth, the definitions of mim-width and one-sided mim-width are} also based on the size of a maximum induced matching in certain subgraphs of the graph {(however, instead of tree decompositions they are based on branch decompositions)}.
While induced matching treewidth lower-bounds the tree-independence number,
it turns out that it is incomparable with one-sided mim-width (and mim-width).
{Bergougnoux et al.~\cite{MR4657707} proved a result analogous to \cref{thm:fvs} for graphs given with a decomposition with bounded one-sided mim-width.}

{
Another interesting parameter based on induced matchings and that in fact captures all these notions (o-mim-width, mim-width, tree-independence number, and induced matching treewidth) is sim-width.
As already mentioned, bounded induced matching treewidth implies bounded sim-width, but not vice versa; for example, the graphs from \Cref{prop:matching-biclique}, which have unbounded induced matching treewidth, have bounded sim-width (see, e.g.,~\cite[Proposition 3.6]{MR3721445}).
One of the main open problems related to sim-width, first asked in~\cite{MR3721445}, is whether \MWIS is solvable in polynomial time for graphs with bounded sim-width.
If the answer is positive, then results from~\cite{MR4563598,DBLP:journals/corr/abs-2308-05817} imply that the same is true for \textsc{Max Weight Distance-$d$ Packing} for all even $d$.
The drawback however is that a branch decomposition of small sim-width is needed in input, as it is another open problem to determine whether such decompositions can be obtained in polynomial time.

The relationships between all these parameters and some others are summarized in \Cref{fig:parameters}.

\begin{figure}[ht]
\centering
\includegraphics[width = 0.8\textwidth]{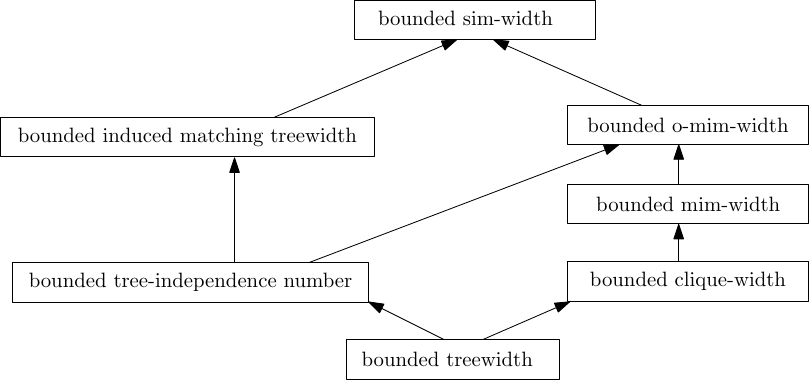}
\caption{Relations between some graph width parameters.}\label{fig:parameters}
\end{figure}

}

\medskip
Recall that classes of bounded tree-independence number are $(\tw,\omega)$-bounded, i.e., treewidth is bounded by a function of the clique number~\cite{dallard2021treewidth}. By \cref{thm:tw-biclique} we know that this is also the case for classes of bounded induced matching treewidth that additionally exclude some fixed biclique as a subgraph.
Furthermore, excluding bicliques is necessary, as $\tw(K_{n,n})=n$ and $\yw(K_{n,n})=1$.
We believe that the following result, which is a common generalization of \cref{thm:delta,thm:tw-biclique}, could also hold.

\begin{conjecture}\label{conj:imt-tin}
For any two integers $k,t \in \mathbb{N}$ there exists an integer $r$ such that each graph with induced matching treewidth at most $k$ and no induced subgraph isomorphic to $K_{t,t}$ has tree-independence number at most $r$.
\end{conjecture}

\medskip
\begin{sloppypar}
The notion of $(\tw,\omega)$-boundedness can be seen as a variant of the extensively studied \hbox{\emph{$\chi$-boundedness}}, where we ask for which classes the chromatic number is upper-bounded by a function of the clique number~\cite{DBLP:journals/jgt/ScottS20}. We believe that graphs of bounded induced matching treewidth are $\chi$-bounded.
\end{sloppypar}

\begin{conjecture}
For any two integers $k,c \in \mathbb{N}$ there exists an integer $r$ such that each graph with induced matching treewidth at most $k$ and clique number at most $c$ has chromatic number at most $r$.
\end{conjecture}

\medskip
Finally, we believe that it would be interesting to explore for what natural classes of graphs induced matching treewidth is bounded. We remark that such results are known for the closely related tree-independence number~\cite{abrishami2023tree,dallard2022secondpaper}.

\end{document}

%% file: commands.tex
\usepackage[utf8]{inputenc}
\usepackage[T1]{fontenc}
\usepackage{amsfonts,amssymb,mathtools}
\usepackage{amsmath,amsthm,xcolor}
\usepackage{hyperref,cite,fullpage}
\hypersetup{
	colorlinks=true,
    linkcolor={red!50!black},
    citecolor={blue!50!black},
    bookmarksopen=true,
	  bookmarksnumbered,
	  bookmarksopenlevel=2,
	  bookmarksdepth=3
}
\usepackage[capitalise]{cleveref}
\usepackage{xspace}
\usepackage{enumerate}
\usepackage{calc}
\usepackage{halloweenmath}
\usepackage{thmtools}
\usepackage{thm-restate}
\usepackage{tikz}
\usetikzlibrary{arrows.meta,patterns}
\usetikzlibrary{ipe} 

\tikzstyle{ipe stylesheet} = [
  ipe import,
  even odd rule,
  line join=round,
  line cap=butt,
  ipe pen normal/.style={line width=0.4},
  ipe pen heavier/.style={line width=0.8},
  ipe pen fat/.style={line width=1.2},
  ipe pen ultrafat/.style={line width=2},
  ipe pen normal,
  ipe mark normal/.style={ipe mark scale=3},
  ipe mark large/.style={ipe mark scale=5},
  ipe mark small/.style={ipe mark scale=2},
  ipe mark tiny/.style={ipe mark scale=1.1},
  ipe mark normal,
  /pgf/arrow keys/.cd,
  ipe arrow normal/.style={scale=7},
  ipe arrow large/.style={scale=10},
  ipe arrow small/.style={scale=5},
  ipe arrow tiny/.style={scale=3},
  ipe arrow normal,
  /tikz/.cd,
  ipe arrows, 
  <->/.tip = ipe normal,
  ipe dash normal/.style={dash pattern=},
  ipe dash dotted/.style={dash pattern=on 1bp off 3bp},
  ipe dash dashed/.style={dash pattern=on 4bp off 4bp},
  ipe dash dash dotted/.style={dash pattern=on 4bp off 2bp on 1bp off 2bp},
  ipe dash dash dot dotted/.style={dash pattern=on 4bp off 2bp on 1bp off 2bp on 1bp off 2bp},
  ipe dash normal,
  ipe node/.append style={font=\normalsize},
  ipe stretch normal/.style={ipe node stretch=1},
  ipe stretch normal,
  ipe opacity 10/.style={opacity=0.1},
  ipe opacity 30/.style={opacity=0.3},
  ipe opacity 50/.style={opacity=0.5},
  ipe opacity 75/.style={opacity=0.75},
  ipe opacity opaque/.style={opacity=1},
  ipe opacity opaque,
]

\definecolor{red}{rgb}{1,0,0}
\definecolor{blue}{rgb}{0,0,1}
\definecolor{green}{rgb}{0,1,0}
\definecolor{yellow}{rgb}{1,1,0}
\definecolor{orange}{rgb}{1,0.647,0}
\definecolor{gold}{rgb}{1,0.843,0}
\definecolor{purple}{rgb}{0.627,0.125,0.941}
\definecolor{gray}{rgb}{0.745,0.745,0.745}
\definecolor{brown}{rgb}{0.647,0.165,0.165}
\definecolor{navy}{rgb}{0,0,0.502}
\definecolor{pink}{rgb}{1,0.753,0.796}
\definecolor{seagreen}{rgb}{0.18,0.545,0.341}
\definecolor{turquoise}{rgb}{0.251,0.878,0.816}
\definecolor{violet}{rgb}{0.933,0.51,0.933}
\definecolor{darkblue}{rgb}{0,0,0.545}
\definecolor{darkcyan}{rgb}{0,0.545,0.545}
\definecolor{darkgray}{rgb}{0.663,0.663,0.663}
\definecolor{darkgreen}{rgb}{0,0.392,0}
\definecolor{darkmagenta}{rgb}{0.545,0,0.545}
\definecolor{darkorange}{rgb}{1,0.549,0}
\definecolor{darkred}{rgb}{0.545,0,0}
\definecolor{lightblue}{rgb}{0.678,0.847,0.902}
\definecolor{lightcyan}{rgb}{0.878,1,1}
\definecolor{lightgray}{rgb}{0.827,0.827,0.827}
\definecolor{lightgreen}{rgb}{0.565,0.933,0.565}
\definecolor{lightyellow}{rgb}{1,1,0.878}
\definecolor{black}{rgb}{0,0,0}
\definecolor{white}{rgb}{1,1,1}

\usetikzlibrary{calc}
\usepackage{lineno}

\newcommand{\cC}{\mathcal{C}}

\newcommand{\cH}{\mathcal{H}}

\newcommand{\cT}{\mathcal{T}}

\newcommand{\tw}{\mathrm{tw}}

\newcommand{\yw}{\mathsf{tree} \textnormal{-} \mu}
\newcommand{\ywT}{\mu_G}

\newcommand{\tin}{\mathsf{tree} \textnormal{-} \alpha}

\newcommand{\N}{\mathbb{N}}
\newcommand{\Q}{\mathbb{Q}}

\newcommand{\MWIS}{\textsc{MWIS}\xspace}

\newcommand{\HH}{\mathcal{H}}
\renewcommand{\P}{\textsf{P}\xspace}
\newcommand{\NP}{\textsf{NP}\xspace}

\newcommand{\FPT}{\textsf{FPT}\xspace}
\newcommand{\Wone}{\textsf{W[1]}\xspace}

\renewcommand{\leq}{\leqslant}
\renewcommand{\geq}{\geqslant}
\renewcommand{\le}{\leqslant}
\renewcommand{\ge}{\geqslant}

\newcommand{\wei}{\mathfrak{w}}
\newcommand{\skel}[1]{\ensuremath{\skull({#1})}}
\newcommand{\limb}[1]{\ensuremath{\dddot{\text{\textleaf}}({#1})}}

\newcommand{\Bag}{X}

\renewcommand{\root}{\mathsf{root}}
\renewcommand{\O}{\mathcal{O}}
\newcommand{\Oh}{\O}
\newcommand{\dist}{\textsf{dist}}

\newcommand{\forked}[1]{\Breve{#1}}
\newcommand{\msotwo}{$\mathsf{MSO}_2$\xspace}
\newcommand{\cmsotwo}{$\mathsf{CMSO}_2$\xspace}
\newcommand{\cpmsotwo}{$\mathsf{C}_{\leq p}\mathsf{MSO}_2$\xspace}

\newcommand{\msotypes}{\mathsf{Types}}
\newcommand{\msotype}{\mathsf{type}}
\newcommand{\dom}{\mathrm{dom}}
\newcommand{\Tab}{\mathsf{Tab}\xspace}

\newcommand{\misets}{\mathbb{I}\xspace}
\newcommand{\FF}{\mathbb{F}\xspace}
\newtheorem{theorem}{Theorem}[section]
\newtheorem{claim}{Claim}[theorem]
\newtheorem{corollary}[theorem]{Corollary}
\newtheorem{proposition}[theorem]{Proposition}
\newtheorem{lemma}[theorem]{Lemma}
\newtheorem{conjecture}[theorem]{Conjecture}
\newtheorem{observation}[theorem]{Observation}

\theoremstyle{definition}
\newtheorem{remark}[theorem]{Remark}

\crefname{claim}{Claim}{Claims}

\usepackage{phaistos}

\newenvironment{claimproof}[1][\unskip]{\noindent {\emph{Proof of Claim #1.\space}}}{\hfill$\triangleleft$ \smallskip}

\usepackage{todonotes}

\usepackage{framed}
\usepackage{tabularx}

\newlength{\RoundedBoxWidth}
\newsavebox{\GrayRoundedBox}
\newenvironment{GrayBox}[1]%
   {\setlength{\RoundedBoxWidth}{.93\columnwidth}
    \def\boxheading{#1}
    \begin{lrbox}{\GrayRoundedBox}
       \begin{minipage}{\RoundedBoxWidth}}%
   {   \end{minipage}
    \end{lrbox}
    \begin{center}
    \begin{tikzpicture}%
       \node(Text)[draw=black!20,fill=white,rounded corners,inner sep=2ex,text width=\RoundedBoxWidth]
             {\usebox{\GrayRoundedBox}};
        \coordinate(x) at (current bounding box.north west);
        \node [draw=white,rectangle,inner sep=3pt,anchor=north west,fill=white]
        at ($(x)+(6pt,.75em)$) {\boxheading};
    \end{tikzpicture}
    \end{center}}

\newenvironment{defproblemx}[1]{\noindent\ignorespaces%
                                \FrameSep=6pt%
                                \parindent=0pt%
                \begin{GrayBox}{#1}%
                \begin{tabular*}{\columnwidth}{!{\extracolsep{\fill}}@{\hspace{.1em}} >{\itshape} p{1.5cm} p{0.86\columnwidth} @{}}%
            }{
                \end{tabular*}%
                \end{GrayBox}%
                \ignorespacesafterend
            }

\newcommand{\problemTask}[3]{%
  \begin{defproblemx}{#1}
    Input: & #2 \\
    Task: & #3
  \end{defproblemx}
}